\pdfoutput=1
\RequirePackage{ifpdf}
\ifpdf 
\documentclass[pdftex]{sigma}
\else
\documentclass{sigma}
\fi

\usepackage{mathtools}
\numberwithin{equation}{section}

\newtheorem{Theorem}{Theorem}[section]
\newtheorem*{Theorem*}{Theorem}
\newtheorem{Corollary}[Theorem]{Corollary}
\newtheorem{Lemma}[Theorem]{Lemma}
\newtheorem{Proposition}[Theorem]{Proposition}

\theoremstyle{definition}
\newtheorem{Definition}[Theorem]{Definition}

\newtheorem{Example}[Theorem]{Example}
\newtheorem{Remark}[Theorem]{Remark}

\begin{document}

\allowdisplaybreaks

\newcommand{\arXivNumber}{2605.09062}

\renewcommand{\PaperNumber}{108}

\FirstPageHeading

\ShortArticleName{Conformal Relative Equilibria on Poisson Manifolds}

\ArticleName{Scaling Symmetries and Conformal Relative\\ Equilibria on Poisson Manifolds,\\ with Applications to Lie--Poisson Systems}

\Author{Manuele SANTOPRETE}

\AuthorNameForHeading{M.~Santoprete}

\Address{Department of Mathematics, Wilfrid Laurier University, Waterloo ON, Canada}
\Email{\mail{msantoprete@wlu.ca}}

\ArticleDates{Received May 15, 2026, in final form October 24, 2026; Published online November 04, 2026}

\Abstract{We investigate conformal relative equilibria for Hamiltonian systems on exact Poisson manifolds equipped with scaling symmetries. By introducing conformally Poisson actions and conformal momentum maps, we characterize these equilibria through an augmented Hamiltonian formulation; in the nondegenerate case, this recovers the conditions recently developed for the exact symplectic case. Specializing to the Lie--Poisson manifold~$\mathfrak g^*$ associated with an arbitrary finite-dimensional real Lie algebra~$\mathfrak g$, we show that its natural scaling action makes~$\mathfrak g^*$ an exact Poisson manifold. We then establish a purely algebraic existence criterion: there exists a homogeneous Hamiltonian on~$\mathfrak g^*$ admitting a~nontrivial conformal relative equilibrium if and only if~$\mathfrak g$ contains an element whose adjoint operator has a nonzero real eigenvalue. Applying this criterion to the Bianchi classification, we identify exactly the real three-dimensional Lie algebras for which there exists a homogeneous Hamiltonian admitting a nontrivial conformal relative equilibrium. As a prominent application, we show that nontrivial conformal relative equilibria exist for a rigid body on~$\mathfrak{so}(2,1)^*$, whereas the classical free rigid body on~$\mathfrak{so}(3)^*$ admits no nontrivial conformal relative equilibria.}


\Classification{53D17; 37J06; 53D20; 70G45; 70H05}

\section{Introduction}

In Hamiltonian mechanics, continuous symmetries play a fundamental role in
the analysis of dynamics, leading to conserved quantities through Noether's
theorem and to reduced descriptions of the phase space through geometric
reduction; see, for example,
\cite{abraham2008foundations,arnold1989mathematical,holm2009geometric,marsden1992lectures,marsden2013introduction}. A central
class of symmetry-generated motions is formed by \textit{relative equilibria},
namely solutions whose trajectories are entirely generated by a one-parameter
subgroup of the symmetry group. These motions appear as steady rotations,
precessions, or other rigidly organized evolutions, and they occupy a central
place in geometric mechanics and stability theory
\cite{marsden1992lectures,SimoLewisMarsden1991}.

While classical relative equilibria are associated with symmetries preserving
the Hamiltonian and the symplectic form, many important mechanical systems
also possess \textit{scaling symmetries}. Such symmetries arise, for instance,
in Kepler-type and $n$-body problems, where spatial rescaling combined with
time reparametrization leaves the equations of motion invariant---a property classically known as mechanical similarity
\cite{bravetti2023scaling, landau1976mechanics}. This leads to
the notion of a \textit{conformal relative equilibrium} (CRE), in which the
vector field aligns with the infinitesimal generator of a scaling action, so
that the corresponding trajectory evolves by pure expansion or contraction \cite{rastelli2026conformal}.
In the classical $n$-body problem, these motions are of paramount importance, as they characterize central configurations,
homothetic collapse, and asymptotic escape trajectories \looseness=-1 \cite{rastelli2026conformal,saari2005collisions,wintner1941analytical}.

 The geometric reduction of Hamiltonian systems admitting such scaling symmetries has recently garnered significant attention. Notably, it has been shown that symplectic Hamiltonian systems with scaling symmetries can be reduced to contact Hamiltonian systems on a space of one less dimension \cite{bravetti2023scaling}. When scaling symmetries commute with standard symmetries, this reduction process naturally yields Kirillov Hamiltonian systems, including a Kirillov analogue of the Lie--Poisson structure \cite{bravetti2024kirillov}.

Recently, a formal geometric framework for scaling symmetries and conformal relative equilibria was developed in the setting of \textit{exact symplectic manifolds} \cite{rastelli2026conformal}. By introducing conformally symplectic actions and conformal momentum maps, it was shown that conformal relative equilibria can be characterized by a condition relating the differential of a suitably modified, or ``augmented'', Hamiltonian to the Liouville one-form.
However, many fundamental mechanical systems---such as the rigid body, the
heavy top, underwater vehicles, and models from fluid dynamics---are naturally formulated on
\textit{Poisson manifolds}, often arising as Lie--Poisson spaces, rather than globally symplectic spaces; see, for
instance, \cite{arnold1998topological,leonard1997stability,lewis1992heavy,marsden1992lectures,ratiu1980motion,weinstein1983local}.
Because Poisson structures are typically degenerate, the exact symplectic
framework cannot be applied globally without modification.

The primary goal of this paper is to extend the geometric theory of scaling symmetries and conformal relative equilibria introduced in \cite{rastelli2026conformal} to the broader setting of \textit{exact Poisson manifolds}, and to apply this theory to completely classify such equilibria on three-dimensional Lie--Poisson spaces.

The paper is organized as follows. Sections~\ref{section2}--\ref{section4} extend the Hamiltonian framework for scaling symmetries and conformal relative equilibria developed in \cite{rastelli2026conformal} from exact symplectic manifolds to exact Poisson manifolds. In Section~\ref{section2}, we recall the general notion of conformal relative equilibrium from \cite{rastelli2026conformal} and introduce the basic definitions of conformally invariant vector fields and conformally Poisson actions, showing that exact Poisson manifolds provide the natural geometric setting for Hamiltonian scaling symmetries.

In Section~\ref{section3}, we define the conformal momentum map for exact Poisson
manifolds, extending the conformally symplectic construction of
\cite{rastelli2026conformal} to the possibly degenerate setting. A
striking feature of the degenerate case, absent from the symplectic
theory, is the prominent role played by Casimir functions.
In Proposition~\ref{prop:casimir-momentum}, we characterize precisely
when a conformal momentum map is a Casimir function. Assuming that a
conformal momentum map exists, it is a Casimir if and only if the
infinitesimal generator of the scaling action is equal to $cD$, where
$c$ is the conformal degree of the action and $D$ is the Liouville
vector field.
We also provide an explicit example of a
non-Casimir conformal momentum map, illustrating that when this
condition fails, a conformal momentum map need not be a Casimir. We
then verify that, when the Poisson tensor is nondegenerate, our
definitions and results reduce to those of
\cite{rastelli2026conformal}. In Section~\ref{section4}, we show that conformal
relative equilibria can be found as zeros of an augmented conformal
Hamiltonian vector field, recovering the exact symplectic conditions
of \cite{rastelli2026conformal} in the nondegenerate case.\looseness=1

The second half of the paper specializes this general theory to \textit{Lie--Poisson manifolds} and constitutes the principal new application of our framework. In Section~\ref{section5}, we show that the dual~$\mathfrak{g}^*$ of any finite-dimensional real Lie algebra $\mathfrak{g}$, equipped with its natural scaling action, canonically forms an exact Poisson manifold. Because the momentum map here is always a Casimir,
we establish a purely algebraic, coordinate-free existence criterion:
there exists a homogeneous Hamiltonian on $\mathfrak g^*$ admitting a
nontrivial conformal relative equilibrium if and only if the underlying
Lie algebra $\mathfrak g$ contains an element whose adjoint map has a
nonzero real eigenvalue (Proposition~\ref{prop:LP-criterion}).

This elegant criterion immediately explains why the classical free rigid body on $\mathfrak{so}(3)^*$ strictly forbids nontrivial conformal relative equilibria, while the rigid body on the non-compact algebra $\mathfrak{so}(2,1)^*$ admits families of them lying on the Casimir light cone.

Finally, in Section~\ref{section6}, we apply our dimension-independent algebraic
criterion (Proposition~\ref{prop:LP-criterion}) to the Bianchi
classification of real three-dimensional Lie algebras. Although the
criterion applies to arbitrary finite-dimensional Lie algebras, we
restrict attention to dimension three, where the Bianchi scheme gives
an explicit list of all real Lie algebras up to isomorphism. Using the
Ellis--MacCallum formulation of this classification
\cite{bianchi1897sugli,bianchi1928lezioni,ellis1969class,ryan2015homogeneous,yoshida2017rattleback},
we determine precisely which Bianchi types admit a homogeneous
Hamiltonian system with a nontrivial conformal relative equilibrium.
They are types III, IV, V, VI$_h$, VII$_h$ in the appropriate
parameter range, and VIII. The corresponding Lie--Poisson tensors and
Casimir functions for all Bianchi types are collected in
Table~\ref{tab:bianchi-JP}.

\section{Scaling symmetries on exact Poisson manifolds}\label{section2}

In this section, we build the geometric foundations for studying scaling symmetries.
We begin in the general setting of smooth manifolds, establishing what it means for
vector fields and functions to be conformally invariant under a scaling action, and
defining conformal relative equilibria as distinguished points on the manifold.
We then specialize to Hamiltonian systems on Poisson manifolds, showing that
the presence of a conformal scaling symmetry naturally forces the underlying Poisson
manifold to carry an exact structure, making exact Poisson manifolds the natural
geometric setting for the remainder of the paper.

\begin{Definition}[conformally invariant vector field]\label{def:conformally-invariant-vector-field}
Let $M$ be a smooth manifold, and let
$\Phi\colon \mathbb{R}^+ \times M \to M$ denote a smooth action of the
multiplicative group~$\mathbb{R}^+$ on~$M$. A vector field~$X$ on~$M$
is said to be \textit{conformally invariant of degree} $a \in \mathbb{R}$ if
\[
 (\Phi_s)_* X = s^a X \qquad \text{for all } s \in \mathbb{R}^+.
\]
In the special case $a = 0$, the vector field $X$ is simply called
\textit{invariant}.
\end{Definition}

\begin{Remark}
We emphasize a change of convention with respect to \cite{rastelli2026conformal}.
There, conformal invariance of vector fields is formulated using pullbacks, whereas
in the present paper we use the pushforward convention
$(\Phi_s)_*X=s^aX$.
This choice is made for consistency with the tensorial conventions adopted later in
the paper.
\end{Remark}

\begin{Definition}[conformally invariant function]
Let $M$ be a smooth manifold, let $\Phi\colon \mathbb{R}^+ \times M \to M$
denote a smooth action of the multiplicative group $\mathbb{R}^+$ on $M$,
and let $F\colon M \to \mathbb{R}$ be a~smooth function.
The function~$F$ is said to be \textit{conformally invariant of degree}
$b \in \mathbb{R}$ if
\[
 \Phi_s^* F = s^b F \qquad \text{for all } s \in \mathbb{R}^+.
\]
In the special case $b = 0$, the function $F$ is simply called \textit{invariant}.
\end{Definition}

With these invariance properties established, we can now formalize the notion of a
conformal relative equilibrium. Rather than viewing them as trajectories, we
characterize them geometrically as specific points where the dynamical vector field
aligns with the symmetry generator. This is one of several equivalent
characterizations established in the conformal relative equilibrium theorem of~\cite{rastelli2026conformal}; we refer the reader to that work for the
full statement and proof.

\begin{Definition}[conformal relative equilibrium]\label{def:conformal-relative-equilibrium}
Let $ M $ be a smooth manifold and $ \Phi\colon  \mathbb{R} ^{ + } \times M \to M $ a smooth action of the multiplicative group $ \mathbb{R} ^{ + } $ on $ M $. Let $ X $ be a vector field on $ M $ that is conformally invariant of degree $ a $. A point $ z _e \in M $ is \textit{called a conformal relative equilibrium} if there exists $ \xi$ in the Lie algebra $ \mathbb{R} $ of the symmetry group $ \mathbb{R} ^{ + } $ such that
$X (z _e) = \xi _M (z _e)$,
where~$ \xi _M $ is the infinitesimal generator of the group action on~$ M $ corresponding to~$ \xi $.
\end{Definition}

Conformal relative equilibria also admit a natural interpretation in terms of reduction. Indeed, if the $\mathbb{R}^{+}$-action is free and proper, the quotient is a smooth manifold, but a conformally invariant vector field does not in general descend to a vector field on that quotient, since it is not invariant under the action. Under suitable assumptions, however, one may rescale the original vector field by a distinguished positive function so as to obtain an invariant vector field, which then defines a reduced dynamics on the quotient. From this viewpoint, conformal relative equilibria are in one-to-one correspondence with equilibria of the conformally reduced system; see  \cite[Section~2]{rastelli2026conformal}.

The preceding definitions apply to arbitrary smooth manifolds. We now restrict our attention to the case where $M$ carries additional geometric structure, specifically an exact Poisson manifold endowed with a smooth action $\Phi\colon \mathbb{R}^{+} \times M \to M$ of the multiplicative group $\mathbb{R}^{+}$.

Recall that a Poisson manifold $(M,\Pi)$ is called \textit{exact} if there exists a vector field $D$ such that
$\mathcal{L}_D \Pi = -\Pi$.
The triple $(M,\Pi,D)$ is called an \textit{exact Poisson manifold}, and the vector field $D$ is called the \textit{Liouville vector field}. Furthermore, on any Poisson manifold, every smooth function determines a Hamiltonian vector field, which in turn governs the dynamics of a Hamiltonian system.

\begin{Definition}[Hamiltonian vector field and Hamiltonian system]\label{def:hamiltonian-system}
Let $(M,\Pi)$ be a Poisson manifold, and let $H\colon M\to \mathbb{R}$ be a smooth function.
The Poisson tensor $\Pi$ induces a bundle map $\Pi^\sharp\colon T^*M\to TM$
characterized by $\beta\bigl(\Pi^\sharp(\alpha)\bigr)=\Pi(\beta,\alpha)$
for all $\alpha,\beta\in T^*M$.
The \textit{Hamiltonian vector field} associated with $H$ is then defined by
$X_H=\Pi^\sharp({\rm d}H)$.
A \textit{Hamiltonian system} on $(M,\Pi)$ is a Poisson manifold together with a
Hamiltonian function~$H$, and its trajectories are the integral curves of~$X_H$,
that is, the solutions of
$\dot z = X_H(z)$.
If $(M,\Pi,D)$ is an exact Poisson manifold, we will refer to $(M,\Pi,D,H)$ as a
\textit{Hamiltonian system on the exact Poisson manifold} $(M,\Pi,D)$.
\end{Definition}

To properly define a scaling symmetry for such a system, we must first introduce a compatibility condition between the $\mathbb{R}^+$-action and the Poisson structure itself.

\begin{Definition}[conformally invariant Poisson tensor]
Let $(M, \Pi, D)$ be an exact Poisson manifold, and let
$\Phi\colon \mathbb{R}^+ \times M \to M$ denote a smooth action of the
multiplicative group $\mathbb{R}^+$ on $M$. The Poisson tensor $\Pi$
is said to be \textit{conformally invariant of degree} $c \in \mathbb{R}$
if, for every $s \in \mathbb{R}^+$, the pushforward under $\Phi_s$ satisfies
$(\Phi_s)_* \Pi = s^{c} \Pi$.
In this case, the action $\Phi$ is called a~\textit{conformally Poisson action}
of degree~$c$.
\end{Definition}

Having specified how the scaling action interacts with the Poisson tensor, we can now define what it means for the full Hamiltonian system to admit a scaling symmetry.

\begin{Definition}[scaling symmetry of a Hamiltonian system]\label{def:scaling-symmetry}
Let $(M,\Pi,D)$ be an exact Poisson manifold and let $H\colon M \to \mathbb{R}$ be a Hamiltonian function. Let
$\Phi\colon \mathbb{R}^{+} \times M \to M$
be a smooth action of the multiplicative group $\mathbb{R}^{+}$ on $M$.
We say that the Hamiltonian system $(M,\Pi,D,H)$ admits a \textit{scaling symmetry} (with respect to $\Phi$) if the Poisson tensor and the Hamiltonian are both conformally invariant under this action, that is,
 \[
 (\Phi_s)_* \Pi = s^c \Pi \qquad \text{and} \qquad \Phi_s^* H = s^b H \qquad \text{for all } s \in \mathbb{R}^{+}.
 \]
\end{Definition}

A remarkable consequence of this setup is that the scaling condition on the Poisson tensor automatically guarantees exactness. The following proposition establishes that exact Poisson manifolds naturally emerge as the appropriate geometric framework for these symmetries, as the infinitesimal generator of the scaling action directly determines the Liouville vector field.

\begin{Proposition}[scaling symmetry and exactness]\label{prop:scaling-exactness}	 Let $(M, \Pi)$ be a Poisson manifold equipped with an $\mathbb{R}^+$-action $\Phi\colon \mathbb{R}^+ \times M \to M$. Suppose there exists a constant $c \in \mathbb{R} \setminus \{0\}$ such that the Poisson tensor $\Pi$ satisfies the scaling condition $(\Phi_s)_* \Pi = s^c \Pi$ for all $s \in \mathbb{R}^+$. Then the Poisson manifold is exact; that is, there exists a vector field $D \in \mathfrak{X}(M)$ such that
$\mathcal{L}_D \Pi = -\Pi$.
\end{Proposition}

\begin{proof}
Using the canonical identification $\mathrm{Lie}(\mathbb{R}^+)\cong \mathbb{R}$,
the exponential map
$\exp\colon \mathbb{R}\to \mathbb{R}^+$
is the usual exponential, given by $\exp(t)={\rm e}^t$. Thus the one-parameter subgroup of $\mathbb{R}^+$ generated by $1\in \mathbb{R}$ is $t\mapsto {\rm e}^t$, and it induces a one-parameter family of diffeomorphisms of $M$ given by
$\varphi_t:=\Phi_{{\rm e}^t}$.

Let $Z$ denote the infinitesimal generator of this flow, namely
\[
Z(x)=\left.\frac{{\rm d}}{{\rm d}t}\right|_{t=0}\varphi_t(x)
=\left.\frac{{\rm d}}{{\rm d}t}\right|_{t=0}\Phi_{{\rm e}^t}(x),
\qquad x\in M.
\]

By definition, the Lie derivative of the bivector field $\Pi$ along $Z$ is
\[
\mathcal{L}_Z\Pi=\left.\frac{{\rm d}}{{\rm d}t}\right|_{t=0}\varphi_t^*\Pi.
\]
Since $\Pi$ is a contravariant tensor field, its pullback is defined by
$\varphi_t^*\Pi=\bigl(\varphi_t^{-1}\bigr)_*\Pi$.
Now, since $\varphi_t=\Phi_{{\rm e}^t}$, we have
$\varphi_t^{-1}=\Phi_{{\rm e}^{-t}}$,
and therefore
$\varphi_t^*\Pi=(\Phi_{{\rm e}^{-t}})_*\Pi$.

Using the scaling assumption $(\Phi_s)_*\Pi=s^c\Pi$ with $s={\rm e}^{-t}$, we obtain
\[
\varphi_t^*\Pi=(\Phi_{{\rm e}^{-t}})_*\Pi=\bigl({\rm e}^{-t}\bigr)^c\Pi={\rm e}^{-ct}\Pi.
\]
Differentiating at $t=0$ yields
\[
\mathcal{L}_Z\Pi
=\left.\frac{{\rm d}}{{\rm d}t}\right|_{t=0}{\rm e}^{-ct}\Pi
=-c \Pi.
\]

Since $c\neq 0$, define
$D=\frac{1}{c}Z$.
Then, by linearity of the Lie derivative,
\[
\mathcal{L}_D\Pi
=\frac{1}{c}\mathcal{L}_Z\Pi
=\frac{1}{c}(-c\Pi)
=-\Pi.
\]
Hence $(M,\Pi)$ is an exact Poisson manifold.
\end{proof}

With the scaling behavior of both the Poisson tensor and the Hamiltonian established, we can now deduce the conformal invariance of the resulting Hamiltonian vector field.

\begin{Proposition}[conformal invariance of the Hamiltonian vector field]\label{prop:conformal-XH}
Let $(M,\Pi,D)$ be an exact Poisson manifold, and assume that
$(M,\Pi,D,H)$ admits a scaling symmetry (with respect to the action
$\Phi$), that is,
\[
 \Phi_s^* H = s^b H, \qquad (\Phi_s)_* \Pi = s^c \Pi,
 \qquad \text{for some } b,c \in \mathbb{R} \text{ and all } s \in \mathbb{R}^+.
\]
Then the Hamiltonian vector field $X_H = \Pi^\sharp({\rm d}H)$ is
conformally invariant of degree $a = c - b$ in the sense of
Definition~{\rm \ref{def:conformally-invariant-vector-field}}, namely,
\[
 (\Phi_s)_* X_H = s^{c-b} X_H
 \qquad \text{for all } s \in \mathbb{R}^+.
\]
\end{Proposition}
\begin{proof}
Let $s\in \mathbb{R}^+$. By the naturality of the Poisson bundle map (Lemma~\ref{lemm:naturality-poisson}), we have the following global equality of bundle maps
\[
\bigl((\Phi_s)_*\Pi\bigr)^\sharp = (\Phi_s)_* \circ \Pi^\sharp \circ \Phi_s^*.
\]
We evaluate both sides of this equation on the globally defined $1$-form ${\rm d}H$.
(Pointwise, at a~point $\Phi_s(x)$, this means evaluating the bundle map on the covector ${\rm d}H_{\Phi_s(x)} \in T_{\Phi_s(x)}^*M$.)

Starting with the right-hand side, the pullback commutes with the exterior derivative. Using the assumed conformal invariance of the Hamiltonian, $\Phi_s^*H = s^bH$, we obtain
\[
\Phi_s^*({\rm d}H) = {\rm d}(\Phi_s^*H) = {\rm d}\bigl(s^bH\bigr) = s^b {\rm d}H.
\]
Applying the remaining bundle maps $\Pi^\sharp$ and $(\Phi_s)_*$ yields
\begin{align*}
\bigl((\Phi_s)_* \circ \Pi^\sharp \circ \Phi_s^*\bigr)({\rm d}H)
& = (\Phi_s)_* \bigl( \Pi^\sharp \bigl( s^b {\rm d}H \bigr) \bigr)
 = (\Phi_s)_* \bigl( s^b  \Pi^\sharp  ( {\rm d}H  ) \bigr) \\
 &
 = (\Phi_s)_* \bigl( s^b X_H \bigr)
 = s^b (\Phi_s)_* X_H.
\end{align*}

Next, we evaluate the left-hand side on ${\rm d}H$. Using the assumed scaling symmetry of the Poisson tensor, $(\Phi_s)_*\Pi = s^c\Pi$, we get
\[
((\Phi_s)_*\Pi)^\sharp ({\rm d}H) = (s^c\Pi)^\sharp({\rm d}H) = s^c \Pi^\sharp({\rm d}H) = s^c X_H.
\]

Since the left- and right-hand sides are equal, we have
$s^c X_H = s^b (\Phi_s)_* X_H$.
Multiplying both sides by $s^{-b}$ yields
$(\Phi_s)_* X_H = s^{c-b} X_H$.
Thus, the Hamiltonian vector field $X_H$ is conformally invariant of degree $a=c-b$.
\end{proof}

To ground these abstract definitions in a concrete physical setting, our first example is the free rigid body---one of the most classical systems in Hamiltonian mechanics. Its equations of motion are naturally formulated on the dual $\mathfrak{so}(3)^*$ of the Lie algebra $\mathfrak{so}(3)$, which carries a Lie--Poisson structure (recalled in Section~5). The quadratic Hamiltonian encoding the kinetic energy is homogeneous of degree two with respect to the scaling action, making it an ideal test case for the theory developed above. As we shall see, the geometry of $\mathfrak{so}(3)^*$ forces every conformal relative equilibrium to be an ordinary equilibrium, a fact that will also be recovered from the general Lie-algebraic criterion of Proposition~\ref{prop:LP-criterion}.

\begin{Example}[free rigid body on $\mathfrak{so}(3)^*$]\label{ex:rigid-body}
Consider the exact Poisson manifold $(M,\Pi,D)$, where $M=\mathbb R^3$,
\[
\Pi(x)=
\begin{pmatrix}
0 & x_3 & -x_2\\
-x_3 & 0 & x_1\\
x_2 & -x_1 & 0
\end{pmatrix},
\qquad
D=x_i\partial_{x_i}.
\]
Here $\Pi$ is the Lie--Poisson tensor on $\mathfrak{so}(3)^*\cong\mathbb R^3$, and
$D$ is the Liouville vector field associated with the scaling action
$\Phi_s(x)=s x$, $s>0$.
Since $\Pi$ is linear in $x$, we have
$(\Phi_s)_*\Pi=s \Pi$,
or equivalently
$L_D\Pi=-\Pi$,
which confirms that $(M,\Pi,D)$ is an exact Poisson manifold, and shows that the scaling action is conformally Poisson of degree $c=1$.

 For $\xi\in \mathrm{Lie}(\mathbb R^+)\cong\mathbb R$, the corresponding infinitesimal generator on $M$ is
\[
\xi_M(x)=\left.\frac{{\rm d}}{{\rm d}t}\right|_{t=0}\Phi_{\exp(t\xi)}(x)
=\left.\frac{{\rm d}}{{\rm d}t}\right|_{t=0}{\rm e}^{t\xi}x
=\xi x
=\xi D(x).
\]
Hence, in this case, $\xi_M=\xi D$.

Now let
\[
H(x)=\tfrac12\bigl(\alpha x_1^2+\beta x_2^2+\gamma x_3^2\bigr),
\]
where $\alpha,\beta,\gamma>0$ are pairwise distinct. Since $H$ is homogeneous of degree~$2$, it is conformally invariant of degree $b=2$ with respect to the scaling action. Therefore, by Proposition~\ref{prop:conformal-XH}, the Hamiltonian vector field
$X_H=\Pi^\sharp({\rm d}H)$
is conformally invariant of degree
$a=c-b=1-2=-1$.

The Hamiltonian vector field is $X_H=\Pi^\sharp({\rm d}H)=\Pi\nabla H$. Since
$\nabla H(x)=(\alpha x_1,\beta x_2,\gamma x_3)^{\mathsf T} $, we compute
\[
X_H(x)=\Pi(x)\nabla H(x)
=
\begin{pmatrix}
(\beta-\gamma)x_2x_3\\
(\gamma-\alpha)x_3x_1\\
(\alpha-\beta)x_1x_2
\end{pmatrix}.
\]
Thus
\[
\dot x_1  = (\beta-\gamma)x_2x_3,\qquad
\dot x_2  = (\gamma-\alpha)x_3x_1,\qquad
\dot x_3  = (\alpha-\beta)x_1x_2.
\]
A conformal relative equilibrium at a nonzero point $x_e$ satisfies
$X_H(x_e)=\xi_M(x_e)$
for some $\xi\in \mathbb{R} $. Since $\xi_M=\xi D$, this condition becomes
$X_H(x_e)=\xi D(x_e)=\xi x_e$.

On the other hand,
$X_H(x)=\nabla H(x)\times x$,
so $X_H(x)$ is orthogonal to $x$ for every $x\in\mathbb R^3$. Taking the Euclidean inner product of the conformal relative equilibrium equation with $x_e$ gives
\[
X_H(x_e)\cdot x_e=\xi |x_e|^2.
\]
But also
\[X_H(x_e)\cdot x_e=(\nabla H(x_e)\times x_e)\cdot x_e=0.\]
Hence
$\xi |x_e|^2=0$.
Since $x_e\neq 0$, we conclude that $\xi=0$.

Therefore, every conformal relative equilibrium is in fact an ordinary equilibrium. The equilibrium condition $X_H(x_e)=0$ is equivalent to
\[
(\beta-\gamma)x_2x_3=0,\qquad
(\gamma-\alpha)x_3x_1=0,\qquad
(\alpha-\beta)x_1x_2=0.
\]
Since $\alpha$, $\beta$, $\gamma$ are pairwise distinct, all three coefficients are nonzero, so at least two components of $x_e$ must vanish. Thus the nonzero conformal relative equilibria are precisely
\[
(\lambda,0,0),\qquad (0,\lambda,0),\qquad (0,0,\lambda),
\qquad \lambda\neq 0,
\]
corresponding to steady rotations about the three principal axes of the rigid body. See Figure~\ref{fig:so3}.
\end{Example}

\begin{figure}[ht]
\centering
\includegraphics[width=0.6\textwidth]{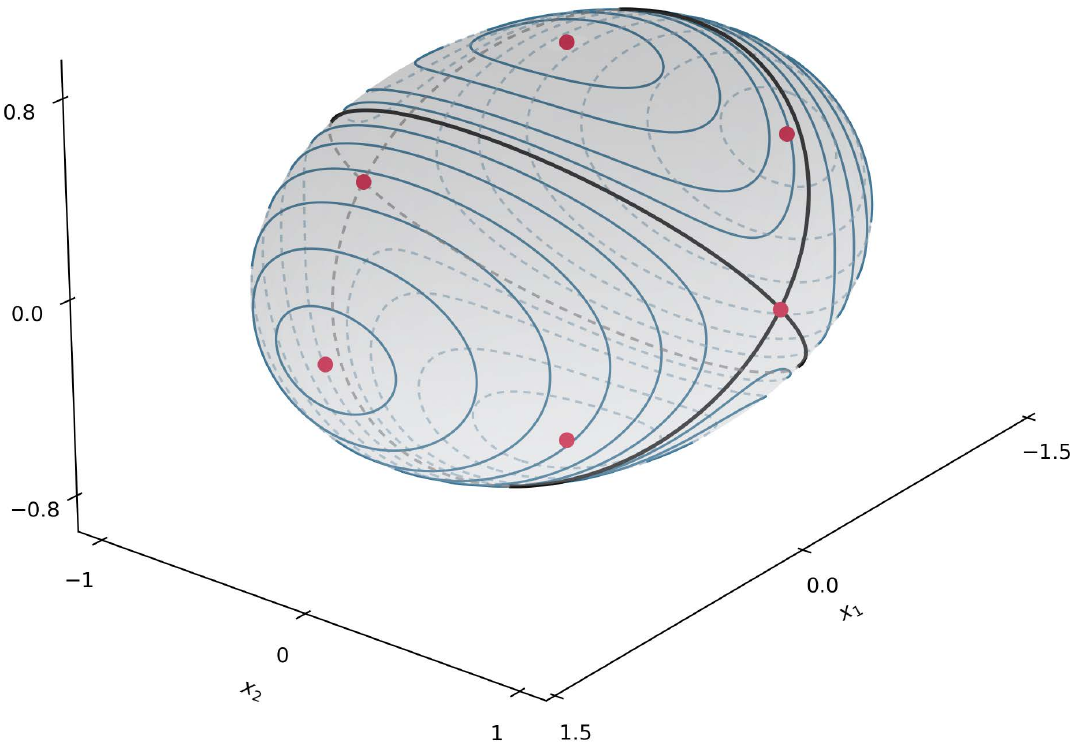}
\caption{Free rigid body on $\mathfrak{so}(3)^*$ for
$H(x)=\tfrac12\bigl(x_1^2+2x_2^2+3x_3^2\bigr)$.
The translucent surface is the energy ellipsoid $H=1$. The blue curves are the intersections $\{H=1\}\cap\{C=c\}$, where
$C(x)=x_1^2+x_2^2+x_3^2$, with hidden back-side portions shown as lighter
dashed curves of the same thickness. The darker highlighted curves correspond
to the special level $C=1$, which passes through the saddle equilibria
$(0,\pm1,0)$, and the red points are the equilibria on the energy level
$H=1$, which in this example coincide with the conformal relative
equilibria.}\label{fig:so3}
\end{figure}

\section{Conformal Hamiltonian systems and momentum maps}\label{section3}

In this section, we introduce the conformally Hamiltonian framework on exact Poisson manifolds endowed with conformal symmetries. The Liouville vector field $D$ allows one to modify the usual Hamiltonian vector field $\Pi^\sharp({\rm d}F)$ by adding a multiple of $D$, leading to the notion of a conformally Hamiltonian system. We then introduce conformal momentum maps for conformally Poisson actions of $\mathbb R^+$.
For the scaling action in Example~\ref{ex:rigid-body}, a conformal momentum map is necessarily a Casimir, as we will verify in Example~\ref{ex:so3-augmented}. This is not true in general on degenerate Poisson manifolds.
We also give a criterion characterizing precisely when a conformal
momentum map is a Casimir function.
Finally, we show that, in the nondegenerate case, the whole framework reduces to the exact symplectic one of \cite{rastelli2026conformal}.

\subsection{The exact Poisson manifold case}\label{section3.1}
We now turn to the exact Poisson setting, where we define conformally Hamiltonian vector fields and the associated conformal momentum maps for conformally Poisson $\mathbb R^+$-actions.

\begin{Definition}[conformally Hamiltonian systems on an exact Poisson manifold]
Let $(M,\Pi,D)$ be an exact Poisson manifold and let $F\in C^\infty(M)$.
A vector field $X_F^\kappa$ is called a {\it conformally Hamiltonian vector field} with parameter $ \kappa \in \mathbb{R} $ if $X_F^\kappa=\Pi^\sharp({\rm d}F)+\kappa D$.

The function $F$ is called the {\it conformal Hamiltonian}, and the
quadruple $(M,\Pi,D,X_F^\kappa)$ is called a {\it conformally Hamiltonian
system} with parameter $\kappa$.

In the special case $\kappa=0$, the vector field reduces to the usual
Hamiltonian vector field
$X_F^0=\Pi^\sharp({\rm d}F)$,
and $\bigl(M,\Pi,D,X_F^0\bigr)$ reduces to a Hamiltonian system on the Poisson
manifold $(M,\Pi)$.
\end{Definition}

\begin{Definition}[conformal momentum map on an exact Poisson manifold]\label{def:conformal-momentum-map}
Let $(M,\Pi,D)$ be an exact Poisson manifold and
$
\Phi\colon \mathbb{R}^+\times M\to M
$
an action of the multiplicative group~$\mathbb{R}^+$ on~$M$.
Suppose that $\Phi$ is a conformally Poisson action of degree
$c\in\mathbb{R}$.
A map $J\colon M\to\mathbb{R}$ is called a~{\it conformal momentum map of
degree $c$} for the action if for every $\xi\in\mathbb{R}$,
\[
\xi_M=\Pi^\sharp({\rm d}J_\xi)+(c\xi)D,
\]
where $J_\xi=\xi J$ and $\xi_M$ denotes the infinitesimal generator
corresponding to $\xi$.

Equivalently,
$X_{J_\xi}^{c\xi}=\xi_M$,
that is, the infinitesimal generator~$\xi_M$ is a conformally
Hamiltonian vector field with parameter $\kappa=c\xi$ and conformal
Hamiltonian~$J_\xi$.
\end{Definition}
The defining equation for the conformal momentum map simplifies significantly when the infinitesimal generator of the action is a scalar multiple of the Liouville vector field. Under this condition, the Hamiltonian vector field of the conformal momentum map vanishes identically, guaranteeing that the momentum map is a Casimir function.

\begin{Proposition}[criterion for Casimir conformal momentum maps]\label{prop:casimir-momentum}
Let $(M,\Pi,D)$ be an exact Poisson manifold equipped with a
conformally Poisson action of degree $c\in\mathbb{R}$. Let $Z$ denote
the infinitesimal generator of the action corresponding to
$1\in\operatorname{Lie}(\mathbb{R}_+)\simeq\mathbb{R}$, and suppose
that the action admits a conformal momentum map $J$. Then the
following are equivalent:
\begin{enumerate}\itemsep=0pt
\item[$(1)$] $J$ is a Casimir function;
\item[$(2)$] $Z=cD$.
\end{enumerate}
\end{Proposition}
\begin{proof}
By Definition~\ref{def:conformal-momentum-map}, the conformal momentum map $J$ satisfies
$Z=\Pi^\sharp({\rm d}J)+cD$.
Hence
$\Pi^\sharp({\rm d}J)=Z-cD$.
Therefore, $J$ is a Casimir if and only if
$\Pi^\sharp({\rm d}J)=0$, which holds if and only if $Z=cD$.
\end{proof}

\begin{Remark}
The preceding criterion separates the Casimir question from the
existence question. If a conformal momentum map $J$ exists, then it is
a Casimir if and only if $Z=cD$. If $Z\neq cD$ at some point, then no
conformal momentum map for the action can be a Casimir. On the other
hand, the existence of a conformal momentum map requires the Poisson
vector field $Z-cD$ to be Hamiltonian, namely to be of the form $\Pi^\sharp({\rm d}J)$. \end{Remark}

\begin{Remark}
Suppose that $c\neq0$ and that the Liouville vector field is chosen as
in Proposition~\ref{prop:scaling-exactness}, namely
$D=\frac{1}{c}Z$,
where $Z$ is the infinitesimal generator of the conformally Poisson
action. Then $Z=cD$. Hence, by
Proposition~\ref{prop:casimir-momentum}, every conformal momentum map
for this action is a Casimir function. In particular, the zero function
$J=0$ is a conformal momentum map.
\end{Remark}

Proposition~\ref{prop:casimir-momentum} shows that a conformal momentum
map is a~Casimir precisely when the infinitesimal generator of the
action is equal to~$cD$. When this condition fails, any conformal
momentum map, if it exists, is necessarily non-Casimir. The following
example illustrates how a shifted scaling action gives rise to such a
non-Casimir conformal momentum map.

\begin{Example}[a non-Casimir conformal momentum map]\label{ex:non-casimir-momentum}
Consider $M = \mathbb{R}^3$ with coordinates $x = (x_1, x_2, x_3)^{\mathsf T}$ and the degenerate Poisson tensor represented by the matrix:
\[
\Pi(x) = \begin{pmatrix}
0 & 1 & 0 \\
-1 & 0 & 0 \\
0 & 0 & 0
\end{pmatrix}.
\]
The coordinate $x_3$ is a global Casimir since $\Pi^\sharp({\rm d}x_3)=0.$
One can verify that $D(x) = (x_1, 0, x_3)^{\mathsf T}$ serves as a Liouville vector field, making $(M, \Pi, D)$ an exact Poisson manifold.

Now, consider the smooth action of $\mathbb{R}^+$ on $M$ given by
\[
\Phi_s(x) = (s(x_1+1)-1, x_2, s x_3)^{\mathsf T}.
\]
The Jacobian matrix of this transformation is the diagonal matrix $\mathcal{J}_s = \operatorname{diag}(s, 1, s)$. Matrix multiplication confirms the pushforward condition $\Phi_{s*} \Pi = \mathcal{J}_s \Pi \mathcal{J}_s^{\mathsf T} = s \Pi$, demonstrating that the action is conformally Poisson of degree $c=1$.

For a Lie algebra element $\xi \in \mathbb{R}$, the corresponding infinitesimal generator of the action is
\[
\xi_M(x) = \frac{{\rm d}}{{\rm d}t}\bigg|_{t=0} \Phi_{{\rm e}^{t\xi}}(x) = \xi \begin{pmatrix} x_1+1 \\ 0 \\ x_3 \end{pmatrix} = \xi \begin{pmatrix} 1 \\ 0 \\ 0 \end{pmatrix} + \xi D(x).
\]
By Definition~\ref{def:conformal-momentum-map}, a conformal momentum map $J$ of degree $c=1$ must satisfy $\xi_M = \Pi^\sharp({\rm d}J_\xi) + \xi D$, where $J_\xi = \xi J$.

Substituting our expression for $\xi_M$ and using the standard
coordinate identification of ${\rm d}J$ with the column vector $\nabla J$,
the vector field $\Pi^\sharp({\rm d}J)$ is represented by $\Pi(x)\nabla J$.
Thus,
\[
\xi \begin{pmatrix} 1 \\ 0 \\ 0 \end{pmatrix} + \xi D(x) = \xi \Pi \nabla J + \xi D(x) \implies \Pi \nabla J = \begin{pmatrix} 1 \\ 0 \\ 0 \end{pmatrix}.
\]
Expanding the matrix multiplication on the left gives
\[
\begin{pmatrix} 0 & 1 & 0 \\ -1 & 0 & 0 \\ 0 & 0 & 0 \end{pmatrix} \begin{pmatrix} \partial_{x_1} J \\ \partial_{x_2} J \\ \partial_{x_3} J \end{pmatrix} = \begin{pmatrix} \partial_{x_2} J \\ -\partial_{x_1} J \\ 0 \end{pmatrix}.
\]
Equating this to $(1, 0, 0)^{\mathsf T}$ requires $\partial_{x_2}J = 1$ and $\partial_{x_1}J = 0$. Therefore, we may choose the conformal momentum map $J(x)=x_2$.

Crucially, $J$ is \emph{not} a Casimir function. Indeed,
\[
\Pi^\sharp({\rm d}J)=\Pi\nabla J=\begin{pmatrix}1\\0\\0\end{pmatrix}\neq 0,
\]
so ${\rm d}J$ does not lie in the kernel of the Poisson tensor. Thus this example shows that, on a~degenerate Poisson manifold, a conformal momentum map need not be a Casimir.
\end{Example}

\subsection{Relationship with the exact symplectic manifold case}\label{section3.2}

In this subsection, we relate the Poisson framework developed above to the exact symplectic setting studied in \cite{rastelli2026conformal}. In that work, conformal actions and momentum maps were formulated in terms of an exact symplectic structure $(M,\omega,\theta)$. Our goal here is to show that, in the nondegenerate case, the Poisson notions introduced in this paper naturally recover those symplectic constructions.

More precisely, we show that an exact Poisson manifold with nondegenerate Poisson tensor canonically induces an exact symplectic structure, and that both conformal actions and conformal momentum maps descend to their symplectic counterparts. We begin by recalling the characterization of exact symplectic manifolds in terms of Liouville vector fields.

\begin{Lemma}[characterization of exact symplectic manifolds]\label{lem:exact-symplectic-manifold}
Let $(M, \omega)$ be a symplectic manifold. The following are equivalent:
\begin{enumerate}\itemsep=0pt
 \item[$1.$] $(M, \omega)$ is an exact symplectic manifold; that is, there exists a $1$-form $\theta \in \Omega^1(M)$ such that $\omega = -{\rm d}\theta$.
 \item[$2.$] There exists a vector field $D \in \mathfrak{X}(M)$, called the Liouville vector field, such that $\mathcal{L}_D \omega = \omega$.
\end{enumerate}
\end{Lemma}

\begin{proof}
$(1) \Rightarrow (2)$: Suppose $\omega = -{\rm d}\theta$. Since the symplectic form $\omega$ is nondegenerate, the bundle map $\omega^\flat\colon TM \to T^*M$ defined by $X \mapsto i_X \omega$ is an isomorphism. Consequently, for the $1$-form $-\theta$, there exists a unique vector field $D$ such that $\omega^\flat(D) = i_D \omega = -\theta$.

Using Cartan's magic formula and the fact that ${\rm d}\omega = 0$ for any symplectic form, we have
\[
\mathcal{L}_D \omega = {\rm d}(i_D \omega) + i_D({\rm d}\omega) = {\rm d}(-\theta) + 0 = -{\rm d}\theta.
\]
By our initial assumption $\omega = -{\rm d}\theta$, it follows that $\mathcal{L}_D \omega = \omega$.

$(2) \Rightarrow (1)$: Conversely, suppose there exists a vector field $D$ such that $\mathcal{L}_D \omega = \omega$. Applying Cartan's formula again,
\[
\omega = \mathcal{L}_D \omega = {\rm d}(i_D \omega) + i_D({\rm d}\omega).
\]
Since ${\rm d}\omega = 0$, this simplifies to $\omega = {\rm d}(i_D \omega)$. By defining the Liouville $1$-form as $\theta \coloneqq -i_D \omega$, we obtain $\omega = {\rm d}(-\theta) = -{\rm d}\theta$, which confirms that the symplectic form is exact.
\end{proof}

We now show that exact Poisson structures with nondegenerate tensor naturally induce exact symplectic structures.

\begin{Lemma}[nondegenerate exact Poisson manifolds are exact symplectic]
Let $(M,\Pi,D)$ be an exact Poisson manifold, and assume that $\Pi$ is nondegenerate. Let $\omega$ be the symplectic form associated with $\Pi$. Then $(M,\omega)$ is an exact symplectic manifold.
\end{Lemma}

\begin{proof}
Since $\Pi$ is nondegenerate, the bundle maps
\begin{alignat*}{3}
& \Pi^\sharp\colon \ T^*M\to TM,
\qquad&&
\bigl\langle\beta,\Pi^\sharp(\alpha)\bigr\rangle=\Pi(\beta,\alpha)
\qquad \text{for all }\alpha,\beta\in T^*M,&\\
& \omega^\flat\colon \ TM\to T^*M,
\qquad &&
\bigl\langle\omega^\flat(X),Y\bigr\rangle=\omega(X,Y)
\qquad \text{for all }X,Y\in TM,&
\end{alignat*}
are mutually inverse. Thus
$
\omega^\flat\circ \Pi^\sharp=\mathrm{Id}_{T^*M}$, $
\Pi^\sharp\circ \omega^\flat=\mathrm{Id}_{TM}.
$
Taking the Lie derivative along $D$ of the first identity gives
\[
\bigl(\mathcal L_D\omega^\flat\bigr)\circ \Pi^\sharp
+
\omega^\flat\circ\bigl(\mathcal L_D\Pi^\sharp\bigr)=0.
\]
Using
$
\mathcal L_D\omega^\flat=(\mathcal L_D\omega)^\flat$,
$\mathcal L_D\Pi^\sharp=(\mathcal L_D\Pi)^\sharp,
$
the previous identity becomes
\[
(\mathcal L_D\omega)^\flat\circ \Pi^\sharp
=
-\omega^\flat\circ(\mathcal L_D\Pi)^\sharp.
\]
Since $\mathcal L_D\Pi=-\Pi$, it follows that
$
(\mathcal L_D\Pi)^\sharp=-\Pi^\sharp,
$
and therefore
\[
(\mathcal L_D\omega)^\flat\circ \Pi^\sharp
=
\omega^\flat\circ\Pi^\sharp.
\]
Composing on the right with $\omega^\flat=(\Pi^\sharp)^{-1}$ yields
\[
(\mathcal L_D\omega)^\flat= \omega^\flat\circ\Pi^\sharp\circ\omega^\flat=\omega^\flat.
\]
Hence $\mathcal L_D\omega=\omega$. By Lemma~\ref{lem:exact-symplectic-manifold}, $(M,\omega)$ is an exact symplectic manifold.
\end{proof}

Having established the correspondence between exact Poisson and exact symplectic structures, we now show that, in the nondegenerate case, the notion of a conformally Poisson action coincides with that of a conformally symplectic action introduced in \cite{rastelli2026conformal}.
\begin{Proposition}[equivalence of conformally Poisson and conformally symplectic actions]
Let $(M,\Pi)$ be a nondegenerate Poisson manifold, and let $\omega$ be the associated symplectic form. Let $\Phi\colon \mathbb{R}^+\times M\to M$ be an action of the multiplicative group $ \mathbb{R} ^{ + } $. Then $\Phi$ is conformally Poisson of degree $c\in\mathbb{R}$, that is,
\[
(\Phi_s)_*\Pi=s^c\Pi \qquad \text{for all } s\in\mathbb{R}^+,
\]
if and only if $\Phi$ is conformally symplectic of degree $c$, that is,
\[
\Phi_s^*\omega=s^c\omega \qquad \text{for all } s\in\mathbb{R}^+.
\]
\end{Proposition}

\begin{proof}
For any diffeomorphism $\Phi_s$, by the naturality of the Poisson bundle map (Lemma \ref{lemm:naturality-poisson}), the pushforward of the bivector $\Pi$ satisfies
\[
((\Phi_s)_* \Pi)^\sharp
= (\Phi_s)_* \circ \Pi^\sharp \circ \Phi_s^* .
\]
Hence the condition
$(\Phi_s)_* \Pi = s^c \Pi$
is equivalent to
\begin{equation}\label{eq:operator-id}
(\Phi_s)_* \circ \Pi^\sharp \circ \Phi_s^* = s^c \Pi^\sharp .
\end{equation}
Taking inverses of both sides of \eqref{eq:operator-id} and using
$
(A\circ B\circ C)^{-1}=C^{-1}\circ B^{-1}\circ A^{-1}
$,
we obtain
\[
(\Phi_s^*)^{-1} \circ \bigl(\Pi^\sharp\bigr)^{-1} \circ ((\Phi_s)_*)^{-1}
= s^{-c} \bigl(\Pi^\sharp\bigr)^{-1}.
\]
Since $\bigl(\Pi^\sharp\bigr)^{-1}=\omega^\flat$, this becomes
\[
\bigl(\Phi_s^{-1}\bigr)^* \circ \omega^\flat \circ \bigl(\Phi_s^{-1}\bigr)_*
= s^{-c} \omega^\flat .
\]
Now replace $s$ by $s^{-1}$. Using the group property of the action,
this gives
\[
\Phi_s^* \circ \omega^\flat \circ (\Phi_s)_*
= s^{c} \omega^\flat .
\]
Therefore, by the naturality of the symplectic bundle map (Lemma~\ref{lemm:naturality-symplectic}) we have
$(\Phi_s^*\omega)^\flat = s^c \omega^\flat$.
Since the map $\eta \mapsto \eta^\flat$ uniquely determines the
$2$-form, it follows that
$\Phi_s^* \omega = s^c \omega$.

The converse implication follows by reversing the argument.
\end{proof}

Finally, we show that, in the nondegenerate exact Poisson setting, the Poisson definition of a conformal momentum map
\[
\xi_M = \Pi^\sharp({\rm d}J_\xi) + (c\xi)D
\]
descends to the symplectic formulation introduced in \cite{rastelli2026conformal}. More precisely, it is equivalent to
\[
i_{\xi_M} \omega = {\rm d}J_\xi - (c\xi)\theta,
\]
where $\theta \coloneqq - i_D\omega$ is the Liouville $1$-form.

\begin{Proposition}[Poisson and symplectic forms of the conformal momentum map]
Let $(M,\Pi,D)$ be an exact Poisson manifold and let $\omega$ be the symplectic form associated with the nondegenerate Poisson tensor $\Pi$. Define the $1$-form $\theta \coloneqq - i_D\omega$. Then the Poisson identity
\[
\xi_M=\Pi^\sharp({\rm d}J_\xi)+(c\xi)D
\]
is equivalent to the symplectic identity
\[
i_{\xi_M}\omega={\rm d}J_\xi-(c\xi)\theta .
\]
\end{Proposition}

\begin{proof}
Apply the bundle map $\omega^\flat\colon TM\to T^*M$ to the identity
$\xi_M=\Pi^\sharp({\rm d}J_\xi)+(c\xi)D$.
Using $\omega^\flat(X)=i_X\omega$ and the linearity of $\omega^\flat$, we obtain
\begin{align*}
i_{\xi_M}\omega
 = \omega^\flat\bigl( \Pi^\sharp({\rm d}J_\xi) \bigr) + \omega^\flat ( (c\xi)D  )
 = \omega^\flat(\Pi^\sharp({\rm d}J_\xi))
+(c\xi)i_D\omega.
\end{align*}
Since $\Pi$ is nondegenerate, $\omega^\flat\circ\Pi^\sharp=\mathrm{Id}_{T^*M}$, hence
$\omega^\flat\bigl(\Pi^\sharp({\rm d}J_\xi)\bigr)={\rm d}J_\xi$.
Therefore,
\[
i_{\xi_M}\omega={\rm d}J_\xi+(c\xi)i_D\omega
= {\rm d}J_\xi-(c\xi)\theta ,
\]
where $\theta=-i_D\omega$.
\end{proof}

\section{Conformal relative equilibria of scaling symmetries}\label{section4}
In this section, we specialize the general notion of conformal relative
equilibria introduced in Definition~\ref{def:conformal-relative-equilibrium}
to the setting of Hamiltonian dynamics on exact Poisson manifolds.
Our goal is to characterize conformal relative equilibria for such systems
in terms of an augmented Hamiltonian and, in the nondegenerate case,
to recover the exact symplectic condition of~\cite{rastelli2026conformal}.
More precisely, we consider a Hamiltonian system $(M,\Pi,D,H)$ on an
exact Poisson manifold $(M,\Pi,D)$, endowed with a conformally Poisson
action $\Phi\colon \mathbb{R}^+ \times M \to M$ of the multiplicative group
$\mathbb{R}^+$ of degree $c$.

We assume that the Hamiltonian $H$ is conformally invariant of degree
$b$ with respect to this action, so that $\Phi$ defines a scaling
symmetry of the Hamiltonian system in the sense of
Definition~\ref{def:scaling-symmetry}:
\[
\Phi_s^* H = s^b H, \qquad (\Phi_s)_* \Pi = s^c \Pi, \qquad s \in \mathbb{R}^+.
\]
Under these assumptions, Proposition~\ref{prop:conformal-XH} implies
that the associated Hamiltonian vector field $X_H = \Pi^\sharp({\rm d}H)$
is conformally invariant of degree $a = c - b$.
We furthermore assume that the action admits an associated conformal
momentum map~$J$ of degree~$c$ in the sense of
Definition~\ref{def:conformal-momentum-map}.
The proposition below shows that conformal relative equilibria are
precisely the zeros of a suitable conformal Hamiltonian vector field
associated with an augmented Hamiltonian.

\begin{Proposition}[augmented Hamiltonian characterization of CRE]\label{prop:augmented-characterization}
Let $(M,\Pi,D)$ be an exact Poisson manifold and let
$\Phi \colon \mathbb{R}^+ \times M \to M$ be a conformally Poisson action
of degree $c \in \mathbb{R}$.
Let $J\colon M \to \mathbb{R}$ be a conformal momentum map of degree $c$
for this action.
For a conformally invariant Hamiltonian $H \in C^\infty(M)$ and a Lie
algebra element $\xi \in \mathbb{R}$, define the \textit{conformal Hamiltonian
vector field} associated with the \textit{augmented Hamiltonian}
$H_\xi = H - J_\xi$ by
\begin{equation}\label{eq:conf-Ham-field}
 X_{H,\xi}^{\mathrm{conf}} \coloneqq \Pi^\sharp({\rm d}H_\xi) - (c\xi) D.
\end{equation}
Then a point $z_e \in M$ is a conformal relative equilibrium if and only if
there exists $\xi \in \mathbb{R}$ such that $z_e$ is a zero of this vector field, that is,
\begin{equation}\label{eq:conf-RE-condition}
 X_{H,\xi}^{\mathrm{conf}}(z_e) = 0.
\end{equation}
\end{Proposition}

\begin{proof}
By Definition~\ref{def:conformal-relative-equilibrium}, $z_e$ is a conformal
relative equilibrium if and only if there exists $\xi \in \mathbb{R}$
such that
$X_H(z_e) = \xi_M(z_e)$,
where $\xi_M$ denotes the infinitesimal generator of the action corresponding
to $\xi$.
Using $X_H = \Pi^\sharp({\rm d}H)$ and, by
Definition~\ref{def:conformal-momentum-map},
$\xi_M = \Pi^\sharp({\rm d}J_\xi) + (c\xi)D$, we obtain
\[
 \Pi^\sharp({\rm d}H)(z_e)
 = \Pi^\sharp({\rm d}J_\xi)(z_e) + (c\xi)D(z_e).
\]
Rearranging all terms to one side and using the linearity of
$\Pi^\sharp$ and of the exterior derivative~${\rm d}$ gives
\[
\Pi^\sharp ({\rm d}(H - J_\xi) )(z_e) - (c\xi)D(z_e) = 0.
\]
Recognizing $H_\xi = H - J_\xi$, the left-hand side is precisely
$X_{H,\xi}^{\mathrm{conf}}(z_e)$ as defined in \eqref{eq:conf-Ham-field},
which yields the condition~\eqref{eq:conf-RE-condition}.
\end{proof}

When $\Pi$ is nondegenerate, the condition of the proposition above
can be reformulated as a condition on the augmented
Hamiltonian $H_\xi$, involving the Liouville one-form $\theta$ and
recovering the exact symplectic result of~\cite{rastelli2026conformal}.

\begin{Corollary}[symplectic characterization of CRE] 
Under the hypotheses of Proposition~{\rm \ref{prop:augmented-characterization}}, assume in addition that $\Pi$ is nondegenerate, let $\omega$ be the associated symplectic form, and set $\theta \coloneqq - i_D\omega$. Then a point $z_e\in M$ is a conformal relative equilibrium with velocity $\xi\in\mathbb{R}$ if and only if
\[
{\rm d}H_\xi(z_e) + (c\xi) \theta(z_e)=0.
\]
In particular, this recovers the exact symplectic condition obtained in~{\rm \cite{rastelli2026conformal}}.
\end{Corollary}

\begin{proof}
From the proposition above, we have
\[
 \Pi^\sharp ({\rm d}(H - J_\xi) )(z_e) = (c\xi)D(z_e).
\]
Since $\Pi$ is nondegenerate, $\Pi^\sharp$ is invertible with
$\bigl(\Pi^\sharp\bigr)^{-1} = \omega^\flat$, so applying $\omega^\flat$ to
both sides yields
\[
{\rm d}(H - J_\xi)(z_e) = (c\xi) i_D\omega(z_e).
\]
Using $\theta = - i_D\omega$ and $H_\xi = H - J_\xi$ gives
\[
{\rm d}H_\xi(z_e) + (c\xi) \theta(z_e) = 0,
\]
as claimed.
\end{proof}

\begin{Remark}
In the special case $c = 0$, the conformal relative equilibrium
condition reduces to
\[
 X_{H_\xi}(z_e) = \Pi^\sharp ({\rm d}H_\xi )(z_e) = 0,
\]
which is the standard augmented Hamiltonian condition for relative
equilibria on a Poisson manifold. When $\Pi$ is nondegenerate,
$\Pi^\sharp$ is invertible, so this is equivalent to
${\rm d}H_\xi(z_e) = 0$, i.e., $z_e$ is a critical point of the augmented
Hamiltonian~$H_\xi$. When $\Pi$ is degenerate, however,
${\rm d}H_\xi(z_e)$ may be a nonzero element of $ \ker \Pi ^{ \sharp} _{ z _e }$,
so $X_{H_\xi}(z_e) = 0$ does not force ${\rm d}H_\xi(z_e) = 0$.
\end{Remark}

\begin{Example}[augmented Hamiltonian for the free rigid body]\label{ex:so3-augmented}
We revisit the free rigid body on $\mathfrak{so}(3)^*$ from Example~\ref{ex:rigid-body} to illustrate the augmented Hamiltonian characterization. The exact Poisson manifold is $M=\mathbb{R}^3$ with the standard Lie--Poisson tensor $\Pi(x)$ and Liouville vector field $D = x_i \partial_{x_i}$. The scaling action $\Phi_s(x)=sx$ is conformally Poisson of degree $c=1$.

For a Lie algebra element $\xi \in \mathbb{R}$, the corresponding infinitesimal generator of the scaling action is
\[
\xi_M(x) = \frac{{\rm d}}{{\rm d}t}\bigg|_{t=0} \Phi_{{\rm e}^{t\xi}}(x) = \xi x = \xi D(x).
\]
Since $\xi_M = c\xi D$ with $c=1$, Proposition~\ref{prop:casimir-momentum} guarantees that any conformal momentum map $J$ for this action is necessarily a Casimir function. We can therefore choose the standard scaled Casimir $J(x) = \frac{1}{2}\bigl(x_1^2 + x_2^2 + x_3^2\bigr)$, which is homogeneous of degree $2$.

For the kinetic energy Hamiltonian $H(x) = \frac{1}{2}\bigl(\alpha x_1^2 + \beta x_2^2 + \gamma x_3^2\bigr)$, the augmented Hamiltonian with parameter $\xi \in \mathbb{R}$ is
\[
H_\xi(x) = H(x) - J_\xi(x) = \frac{1}{2}\bigl((\alpha - \xi)x_1^2 + (\beta - \xi)x_2^2 + (\gamma - \xi)x_3^2\bigr).
\]
By Proposition~\ref{prop:augmented-characterization}, a conformal relative equilibrium $x_e \neq 0$ is a zero of the conformal Hamiltonian vector field. Since $J_\xi$ is a Casimir, its Hamiltonian vector field vanishes ($\Pi^\sharp({\rm d}J_\xi) = 0$), so the conformal vector field simplifies to
\[
X_{H,\xi}^{\mathrm{conf}}(x) = \Pi^\sharp({\rm d}H) - \xi D =
\begin{pmatrix}
(\beta-\gamma)x_2 x_3 - \xi x_1 \\
(\gamma-\alpha)x_1 x_3 - \xi x_2 \\
(\alpha-\beta)x_1 x_2 - \xi x_3
\end{pmatrix}.
\]
Setting $X_{H,\xi}^{\mathrm{conf}}(x_e) = 0$ yields the system
\[
(\beta-\gamma)x_2 x_3 = \xi x_1,\qquad
(\gamma-\alpha)x_1 x_3 = \xi x_2,\qquad
(\alpha-\beta)x_1 x_2 = \xi x_3.
\]
These are exactly the componentwise algebraic conditions derived in Example~\ref{ex:rigid-body}. This confirms that the conformal relative equilibria are precisely the zeros of the augmented Hamiltonian dynamics, gracefully recovering our earlier conclusion that the rigid body admits no nontrivial conformal relative equilibria.
\end{Example}

\section{Lie--Poisson manifolds and scaling symmetries}\label{section5}

In this section, we specialize the general theory developed above to
Lie--Poisson manifolds of the form $(\mathfrak{g}^*, \Pi_{\mathrm{LP}})$,
where $\mathfrak{g}$ is a finite-dimensional real Lie algebra and
$\Pi_{\mathrm{LP}}$ is the Lie--Poisson tensor introduced below. We show
that $\mathfrak{g}^*$ endowed with the natural scaling action is an exact
Poisson manifold, and we obtain a criterion, expressed purely in terms of
the adjoint representation of~$\mathfrak{g}$, for the existence of
nontrivial conformal relative equilibria.

\subsection{The scaling action as an exact Poisson structure}\label{section5.1}

Let $\mathfrak{g}$ be a finite-dimensional real Lie algebra with Lie
bracket $[\cdot,\cdot]\colon \mathfrak{g}\times\mathfrak{g}\to\mathfrak{g}$,
and let $\mathfrak{g}^*$ denote its linear dual. Let
$(e_1,\ldots,e_n)$ be a basis of $\mathfrak{g}$, with Lie brackets
\[
 [e_i,e_j] = C^k_{ij} e_k,
\]
where $C^k_{ij}\in\mathbb{R}$ are the structure constants (summation
over repeated indices is implied). Let
$\langle\cdot,\cdot\rangle\colon \mathfrak{g}^*\times\mathfrak{g}\to\mathbb{R}$
denote the natural pairing, and let $\bigl(e^1,\ldots,e^n\bigr)$ be the dual basis
of $\mathfrak{g}^*$, defined by $\langle e^i,e_j\rangle=\delta^i_j$.
 We use $(\mu_1,\ldots,\mu_n)$ as coordinates on
$\mathfrak{g}^*$, where $\mu_i = \langle\mu,e_i\rangle$ are the
components of $\mu\in\mathfrak{g}^*$ in the dual basis.

The {\it Lie--Poisson tensor} on $\mathfrak{g}^*$ is the bivector field
$\Pi_{\mathrm{LP}}$ with components
\[
 \Pi^{ij}_{\mathrm{LP}}(\mu) = C^k_{ij} \mu_k.
\]
One readily verifies that this satisfies the Jacobi identity, so
$(\mathfrak{g}^*, \Pi_{\mathrm{LP}})$ is a Poisson manifold;
see, e.g.,~\cite{marsden2013introduction}. Note that $\Pi_{\mathrm{LP}}$ is linear in~$\mu$; this reflects the fact that the Lie--Poisson structure is entirely determined by the Lie algebra structure of $\mathfrak{g}$.
 The associated Poisson bracket is
\[
 \{f,g\}(\mu)
 = \Pi_{\mathrm{LP}}^{ij}(\mu)
 \frac{\partial f}{\partial\mu_i} \frac{\partial g}{\partial\mu_j}
 = \bigl\langle\mu, [{\rm d}f(\mu), {\rm d}g(\mu)]\bigr\rangle,
 \qquad f,g\in C^\infty(\mathfrak{g}^*),
\]
where we identify ${\rm d}f(\mu)\in T^*_\mu\mathfrak{g}^*\cong\mathfrak{g}$.
The Hamiltonian vector field of $H\in C^\infty(\mathfrak{g}^*)$ is
\[
 X_H(\mu) = \mathrm{ad}^*_{{\rm d}H(\mu)} \mu,
\]
where $\mathrm{ad}^*_\xi\colon \mathfrak{g}^*\to\mathfrak{g}^*$ denotes the
coadjoint action of $\xi\in\mathfrak{g}$, defined by
$\langle\mathrm{ad}^*_\xi \mu, \eta\rangle =
-\langle\mu,[\xi,\eta]\rangle$ for all $\eta\in\mathfrak{g}$.
Consider the smooth action of the multiplicative group $\mathbb{R}^+$
on $\mathfrak{g}^*$ defined by
\[
 \Phi \colon \ \mathbb{R}^+ \times \mathfrak{g}^* \to \mathfrak{g}^*,
 \qquad
 \Phi_s(\mu) = s\mu.
\]
One readily verifies that $\Phi_1 = \mathrm{id}$,
$\Phi_s \circ \Phi_r = \Phi_{sr}$ for all $r,s\in\mathbb{R}^+$,
and that the action map $(s,\mu)\mapsto s\mu$ is smooth, so
$\Phi$ is a smooth action of $(\mathbb{R}^+,\cdot)$ on $\mathfrak{g}^*$,
which we call the \textit{scaling action}. The infinitesimal generator
corresponding to $1\in\mathrm{Lie}(\mathbb{R}^+)\cong\mathbb{R}$ is
\[
 D=\mu_i \partial_{\mu_i},
\]
characterized by
\[
 D(\mu)=\frac{{\rm d}}{{\rm d}t}\bigg|_{t=0}\Phi_{{\rm e}^t}(\mu)=\mu
\]
for all $\mu\in\mathfrak{g}^*$.

The next proposition shows that the natural scaling action on
$\mathfrak{g}^*$ places the Lie--Poisson manifold
$(\mathfrak{g}^*,\Pi_{\mathrm{LP}})$ within the exact Poisson framework
developed above.

\begin{Proposition}[Lie--Poisson manifolds are exact Poisson]
\label{prop:LP-exact}
Let $\mathfrak{g}$ be a finite-dimensional real Lie algebra, and let
$D$ be the infinitesimal generator of the scaling action
$\Phi\colon \mathbb{R}^+\times\mathfrak{g}^*\to\mathfrak{g}^*$ corresponding to
the Lie algebra element $1\in \mathrm{Lie}(\mathbb{R}^+)$. Then the scaling action
is conformally Poisson of degree $c=1$, that is,
\[
 (\Phi_s)_* \Pi_{\mathrm{LP}} = s \Pi_{\mathrm{LP}}
 \qquad\text{for all } s\in\mathbb{R}^+.
\]
Consequently, $(\mathfrak{g}^*,\Pi_{\mathrm{LP}},D)$ is an exact Poisson
manifold.
\end{Proposition}

\begin{proof}
Let $\nu=\Phi_s(\mu)=s\mu$. To prove that
$(\Phi_s)_*\Pi_{\mathrm{LP}}=s \Pi_{\mathrm{LP}}$, it is enough to evaluate
both sides on arbitrary covectors
$\alpha,\beta\in T_\nu^*\mathfrak{g}^*\cong\mathfrak{g}$.

By definition of the pushforward of a bivector,
\[
 ((\Phi_s)_*\Pi_{\mathrm{LP}} )_\nu(\alpha,\beta)
=
\Pi_{\mathrm{LP},\mu}\bigl(\Phi_s^*\alpha,\Phi_s^*\beta\bigr).
\]
Since $\Phi_s\colon \mathfrak{g}^*\to\mathfrak{g}^*$ is the linear map
$\mu\mapsto s\mu$, its pullback on covectors is multiplication by $s$, so
\[
\Phi_s^*\alpha=s\alpha,
\qquad
\Phi_s^*\beta=s\beta.
\]
Using the defining formula for the Lie--Poisson tensor, we obtain
\[
((\Phi_s)_*\Pi_{\mathrm{LP}})_\nu(\alpha,\beta)
=
\Pi_{\mathrm{LP},\mu}(s\alpha,s\beta)
=
\langle \mu,[s\alpha,s\beta]\rangle
=
s^2\langle \mu,[\alpha,\beta]\rangle.
\]
Since $\nu=s\mu$, this becomes
\[
\bigl((\Phi_s)_*\Pi_{\mathrm{LP}}\bigr)_\nu(\alpha,\beta)
=
s \langle \nu,[\alpha,\beta]\rangle
=
s \Pi_{\mathrm{LP},\nu}(\alpha,\beta).
\]
Therefore, $(\Phi_s)_*\Pi_{\mathrm{LP}}=s \Pi_{\mathrm{LP}}$, so the scaling
action is conformally Poisson of degree $c=1$.

Now let $D$ be the infinitesimal generator of the scaling action corresponding to the Lie algebra element $1\in \mathrm{Lie}(\mathbb{R}^+)$. Since
$(\Phi_s)_*\Pi_{\mathrm{LP}}=s \Pi_{\mathrm{LP}}$ for all $s\in\mathbb{R}^+$,
Proposition~\ref{prop:scaling-exactness} applies with $c=1$ and yields
\[
\mathcal{L}_D\Pi_{\mathrm{LP}}=-\Pi_{\mathrm{LP}}.
\]
Hence $(\mathfrak{g}^*,\Pi_{\mathrm{LP}},D)$ is an exact Poisson manifold.
\end{proof}

\begin{Remark}[conformal momentum maps on Lie--Poisson manifolds]\label{rem:lp-casimir}
Consider the natural scaling action $\Phi_s(x) = sx$ on a Lie--Poisson manifold $\mathfrak{g}^*$. By Proposition~\ref{prop:LP-exact}, this action is conformally Poisson of degree $c=1$, and its infinitesimal generator for $\xi \in \mathbb{R}$ is exactly $\xi_M = \xi D$. Therefore, by Proposition~\ref{prop:casimir-momentum}, any conformal momentum map for this action is necessarily a~Casimir function. As a result, the augmented Hamiltonian always takes the form $H_\xi = H - \xi C$ for some Casimir~$C$. Note that because $\Pi^\sharp({\rm d}C) = 0$, the choice of Casimir does not affect the conformal Hamiltonian vector field $X_{H,\xi}^{\mathrm{conf}} = \Pi^\sharp({\rm d}H) - \xi D$, and thus the set of conformal relative equilibria is entirely independent of this choice.
\end{Remark}

\subsection{A criterion for conformal relative equilibria}\label{section5.2}

We now derive an explicit condition for conformal relative equilibria
on $(\mathfrak{g}^*, \Pi_{\mathrm{LP}})$. Since the infinitesimal
generator of the scaling action corresponding to $\xi\in\mathbb{R}$
is $\xi_M = \xi D$, where $ D $ is the Euler vector field, the conformal relative equilibrium condition
$X_H(\mu_e)=\xi_M(\mu_e)$ of Definition~\ref{def:conformal-relative-equilibrium} becomes
$ \mathrm{ad}^*_{{\rm d}H(\mu_e)} \mu_e = \xi \mu_e$,
where we used the expression $X_H(\mu) = \mathrm{ad}^*_{{\rm d}H(\mu)}\mu$
for the Hamiltonian vector field on $\mathfrak{g}^*$ and the fact that
$D(\mu_e) = \mu_e$. Setting $\zeta = {\rm d}H(\mu_e)\in\mathfrak{g}$, this
says that $\mu_e$ is an eigenvector of $\mathrm{ad}^*_\zeta$ with
eigenvalue~$\xi$. For the equilibrium to be nontrivial we need
$\mu_e\neq 0$ and $\xi\neq 0$; as seen in Example~\ref{ex:rigid-body},
the latter condition can fail---in the case of~$\mathfrak{so}(3)^*$
the orthogonality of $X_H$ forces $\xi=0$ for every $\mu_e\neq 0$,
so that no nontrivial conformal relative equilibria exist. The
following proposition gives a purely Lie-algebraic characterization
of exactly when a nontrivial conformal relative equilibrium can exist.

\begin{Proposition}[criterion for CRE on $\mathfrak{g}^*$]\label{prop:LP-criterion}
Let $\mathfrak{g}$ be a finite-dimensional real Lie algebra, and endow
$\mathfrak{g}^*$ with its Lie--Poisson structure and the scaling
action $\Phi_s(\mu)=s\mu$. The following are equivalent:
\begin{enumerate}\itemsep=0pt
 \item[$(1)$] There exists a homogeneous Hamiltonian $H\in
 C^\infty(\mathfrak{g}^*)$ of some degree $b$
 and a nonzero point $\mu_e\in\mathfrak{g}^*$ that is a
 nontrivial conformal relative equilibrium with velocity
 $\xi\neq 0$.
 \item[$(2)$] There exist $\zeta\in\mathfrak{g}$,
 $\mu\in\mathfrak{g}^*\setminus\{0\}$, and
 $\xi\in\mathbb{R}\setminus\{0\}$ such that
 $\mathrm{ad}^*_\zeta \mu = \xi \mu$.
 \item[$(3)$] There exist $\zeta\in\mathfrak{g}$,
 $\eta\in\mathfrak{g}\setminus\{0\}$, and
 $\xi\in\mathbb{R}\setminus\{0\}$ such that
 $\mathrm{ad}_\zeta \eta = \xi \eta$;
 such a~$\zeta$ is called a~\emph{hyperbolic element}
 of~$\mathfrak{g}$.
\end{enumerate}
Moreover, given $\zeta$, $\mu$, and $\xi$ satisfying condition~{\rm (2)}, the linear Hamiltonian $H_\zeta(\nu)=\langle\nu,\zeta\rangle$ realizes condition~{\rm (1)}, with $\mu_e=\mu$ as a nontrivial conformal relative equilibrium of velocity~$\xi$. This realization lies on the zero-energy level $H_\zeta(\mu_e)=0$.
\end{Proposition}

\begin{proof}
(1)$\Rightarrow$(2).
Suppose $\mu_e\neq 0$ is a nontrivial conformal relative equilibrium
for $H$ with velocity $\xi\neq 0$. Then
\[
 X_H(\mu_e) = \mathrm{ad}^*_{{\rm d}H(\mu_e)} \mu_e = \xi \mu_e,
\]
so condition~(2) holds with $\zeta = {\rm d}H(\mu_e)$ and $\mu = \mu_e$.

(2)$\Rightarrow$(1).
Suppose $\mathrm{ad}^*_\zeta \mu = \xi \mu$ for some
$\zeta\in\mathfrak{g}$, $\mu\neq 0$, and $\xi\neq 0$.
Consider the linear Hamiltonian $H_\zeta(\nu)=\langle\nu,\zeta\rangle$.
Since $H_\zeta$ is linear in $\nu$, its differential is the constant
${\rm d}H_\zeta(\nu)=\zeta$ for all $\nu$, and hence its Hamiltonian vector
field is
\[
 X_{H_\zeta}(\nu)
 = \mathrm{ad}^*_{{\rm d}H_\zeta(\nu)} \nu
 = \mathrm{ad}^*_\zeta \nu.
\]

Evaluating at $\mu$ and using the hypothesis gives
\[
X_{H_\zeta}(\mu)
=
\operatorname{ad}^*_\zeta\mu
=
\xi\mu
=
\xi D(\mu),
\]
so $\mu$ is a nontrivial conformal relative equilibrium with velocity
$\xi$.

It remains to verify the final assertion, namely that this linear
realization places the conformal relative equilibrium on the
zero-energy level. Evaluating $\operatorname{ad}^*_\zeta\mu=\xi\mu$
on $\zeta$ gives
\[\xi \mu(\zeta)=\bigl(\operatorname{ad}^*_\zeta\mu\bigr)(\zeta)=-\mu([\zeta,\zeta])=0.\]
Since $\xi\neq0$, it follows that
$H_\zeta(\mu)=\mu(\zeta)=0.$

(2)$\Leftrightarrow$(3).
By the definition of the infinitesimal coadjoint action,
\[
\langle \operatorname{ad}^*_\zeta\mu,\eta\rangle
=
-\langle\mu,\operatorname{ad}_\zeta\eta\rangle
\qquad
\text{for all }\mu\in\mathfrak g^*,\ \eta\in\mathfrak g.
\]
Thus $\operatorname{ad}^*_\zeta$ is the negative dual of
$\operatorname{ad}_\zeta$. In dual bases, its matrix is
$-(\operatorname{ad}_\zeta)^{\mathsf T}$, and consequently
\[
\operatorname{spec}(\operatorname{ad}^*_\zeta)
=
-\operatorname{spec}(\operatorname{ad}_\zeta).
\]
Therefore, $\operatorname{ad}^*_\zeta$ has a nonzero real eigenvalue if
and only if $\operatorname{ad}_\zeta$ does.
\end{proof}

\begin{Remark}
Condition~(3) provides a purely Lie-algebraic test requiring no
knowledge of the Hamiltonian. For instance, for any
$\zeta=(\zeta_1,\zeta_2,\zeta_3)\in\mathfrak{so}(3)\cong\mathbb{R}^3$,
the adjoint operator $\mathrm{ad}_\zeta(v)=\zeta\times v$ is represented
by the skew-symmetric matrix
\[
\mathrm{ad}_\zeta =
\begin{pmatrix}
0 & -\zeta_3 & \zeta_2 \\
\zeta_3 & 0 & -\zeta_1 \\
-\zeta_2 & \zeta_1 & 0
\end{pmatrix}.
\]
Its characteristic equation is
$\lambda\bigl(\lambda^2+\zeta_1^2+\zeta_2^2+\zeta_3^2\bigr)=0$,
giving eigenvalues
$\bigl\{0,\pm i\sqrt{\zeta_1^2+\zeta_2^2+\zeta_3^2}\bigr\}$.
Since these are never real and nonzero, $\mathfrak{so}(3)$ contains no
hyperbolic elements, and hence no homogeneous Hamiltonian system on
$\mathfrak{so}(3)^*$ admits a nontrivial conformal relative
equilibrium, consistently with Example~\ref{ex:rigid-body}.
\end{Remark}

We now present an example where nontrivial conformal relative equilibria
do exist. The Lie algebra $\mathfrak{so}(2,1)$ is noncompact and, unlike
$\mathfrak{so}(3)$, contains hyperbolic elements; by
Proposition~\ref{prop:LP-criterion} this guarantees the existence of
nontrivial conformal relative equilibria for homogeneous Hamiltonian
systems on $\mathfrak{so}(2,1)^*$. The following example exhibits them
explicitly for a quadratic Hamiltonian.

\begin{figure}[ht]
\centering\small
\centerline{Special case $(\alpha,\beta,\gamma)=(1,3,-2)$, cone $H=0$, CRE lines, and intersections with $C={\rm const}$}
\vspace{1mm}

\includegraphics[width=0.6\textwidth]{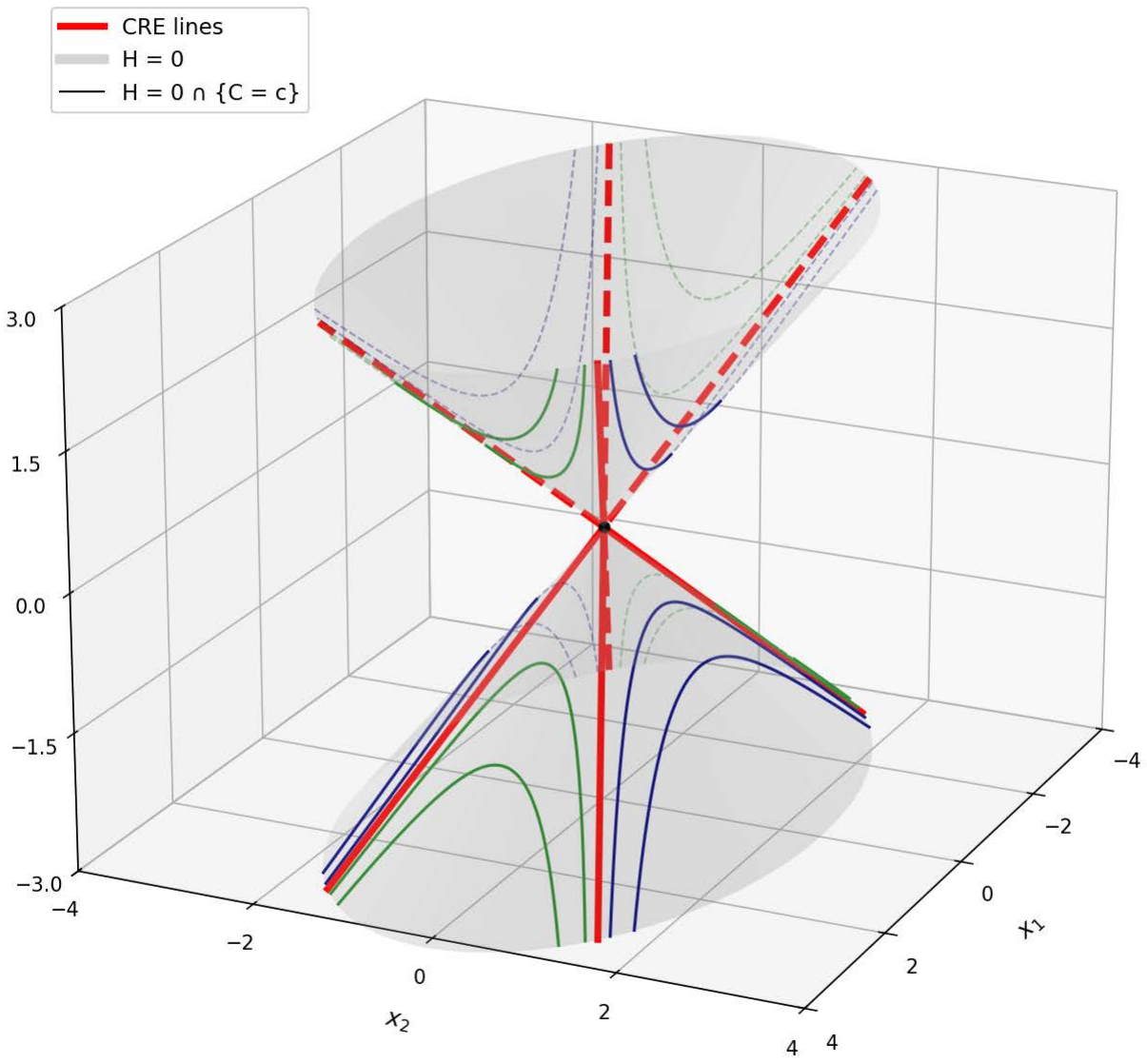}
\caption{Geometry of the nontrivial conformal relative equilibria for the free rigid body on $\mathfrak{so}(2,1)^*$ with Hamiltonian $H(x) = \frac{1}{2}\bigl(x_1^2 + 3x_2^2 - 2x_3^2\bigr)$. The translucent surface is the zero-energy cone $H=0$.
The thick red lines are the eight rays of nontrivial conformal relative
equilibria, which lie exactly in the intersection
$\{H=0\}\cap\{C=0\}$, where
$C(x)=-x_1^2+x_2^2+x_3^2$ is the Casimir function. The thin curves
show the intersections $\{H=0\}\cap\{C=c\}$ for the level sets
$c=\pm1,\pm3$. Portions of the curves and rays lying on the rear of the cone from the viewer's perspective are drawn with dashed lines to indicate depth.}\label{fig:CRE-so(21)}
\end{figure}

\begin{Example}[free rigid body on $\mathfrak{so}(2,1)^*$]\label{ex:so21}
Let $M=\mathbb{R}^3$ with coordinates $x=(x_1,x_2,x_3)$ and the
Lie--Poisson structure of $\mathfrak{so}(2,1)^*$, whose Poisson tensor is
\[
\Pi(x)=
\begin{pmatrix}
0 & x_3 & -x_2\\
-x_3 & 0 & -x_1\\
x_2 & x_1 & 0
\end{pmatrix}.
\]
The Casimir function is
\[
C(x)=-x_1^2+x_2^2+x_3^2.
\]

Consider the quadratic Hamiltonian
\[
H(x)=\tfrac12\bigl(\alpha x_1^2+\beta x_2^2+\gamma x_3^2\bigr),
\]
where $\alpha,\beta,\gamma\in\mathbb{R}$ satisfy
$-\beta<\alpha<-\gamma$.
Equivalently,
$\beta-\gamma>0$, $\alpha+\beta>0$, $\alpha+\gamma<0$.
Since
\[
{\rm d}H(x)=(\alpha x_1,\beta x_2,\gamma x_3),
\qquad
\nabla H(x)=
\begin{pmatrix}
\alpha x_1\\
\beta x_2\\
\gamma x_3
\end{pmatrix},
\]
the Hamiltonian vector field is
\[
X_H(x)=\Pi(x)\nabla H(x)
=
\begin{pmatrix}
(\beta-\gamma)x_2x_3\\
-(\alpha+\gamma)x_1x_3\\
(\alpha+\beta)x_1x_2
\end{pmatrix}.
\]
Equivalently,
\[
\dot{x}_1=(\beta-\gamma)x_2x_3,\qquad
\dot{x}_2=-(\alpha+\gamma)x_1x_3,\qquad
\dot{x}_3=(\alpha+\beta)x_1x_2.
\]

The manifold $M$ admits the scaling action $\Phi_s(x)=sx$ for $s>0$.
Its infinitesimal generator is the Euler vector field
\[
D=x_1\partial_{x_1}+x_2\partial_{x_2}+x_3\partial_{x_3},
\qquad\text{so that}\  D(x)=x.
\]
Since $\Pi$ is linear in $x$, we have $(\Phi_s)_*\Pi=s\Pi$, so the Poisson
tensor is conformally invariant of degree $c=1$. Since $H$ is homogeneous of
degree $b=2$, Proposition~\ref{prop:conformal-XH} implies that $X_H$ is
conformally invariant of degree $a=c-b=-1.$

A nontrivial conformal relative equilibrium at a nonzero point $ x_e$ satisfies
$X_H(x_e)=\xi x_e$
for some $\xi\neq 0$.
Writing this condition componentwise gives
\begin{subequations}\label{eq:so21-cre-system}
\begin{align}
&(\beta-\gamma)x_2x_3= \xi x_1, \label{eq:so21-cre-system-a}\\
&-(\alpha+\gamma)x_1x_3= \xi x_2, \label{eq:so21-cre-system-b}\\
&(\alpha+\beta)x_1x_2= \xi x_3. \label{eq:so21-cre-system-c}
\end{align}
\end{subequations}

None of the coordinates of $x_e$ can vanish. Indeed, if $x_1=0$, then equations \eqref{eq:so21-cre-system-b} and \eqref{eq:so21-cre-system-c}
imply $\xi x_2=0$ and $\xi x_3=0$, hence $x_2=x_3=0$ because $\xi\neq 0$,
contradicting $x_e\neq 0$. The same argument applies to $x_2$ and $x_3$.
Therefore,
$x_1x_2x_3\neq 0$.

We may now eliminate $\xi$. Dividing equation \eqref{eq:so21-cre-system-a} by $x_1$ and equation \eqref{eq:so21-cre-system-b}
by $x_2$, then equating the resulting expressions for~$\xi$, we obtain
\[
(\beta-\gamma)\frac{x_2x_3}{x_1}
=-(\alpha+\gamma)\frac{x_1x_3}{x_2},
\]
and hence
\[
\frac{x_2^2}{x_1^2}
=-\frac{\alpha+\gamma}{\beta-\gamma}.
\]
Similarly, comparing equations \eqref{eq:so21-cre-system-a} and~\eqref{eq:so21-cre-system-c} gives
\[
\frac{x_3^2}{x_1^2}
=\frac{\alpha+\beta}{\beta-\gamma}.
\]
Because $\beta-\gamma>0$, $\alpha+\beta>0$, and $\alpha+\gamma<0$, both
right-hand sides are strictly positive.

Set
\[
p=\sqrt{-\frac{\alpha+\gamma}{\beta-\gamma}},
\qquad
q=\sqrt{\frac{\alpha+\beta}{\beta-\gamma}}.
\]
Then the nontrivial conformal relative equilibria form four punctured lines (equivalently, eight rays):
\[
\mu_e^{\sigma_2,\sigma_3}(t)
=
t
\begin{pmatrix}
1\\
\sigma_2 p\\
\sigma_3 q
\end{pmatrix},
\qquad
t \in \mathbb{R} \setminus \{ 0 \} ,\quad
\sigma_2,\sigma_3\in\{+1,-1\}.
\]

The corresponding velocity is obtained from equation \eqref{eq:so21-cre-system-a}:
\[
\xi
=
\frac{(\beta-\gamma)x_2x_3}{x_1}
=
(\beta-\gamma)\sigma_2\sigma_3 pqt.
\]
Since
\[
pq
=
\frac{\sqrt{-(\alpha+\gamma)(\alpha+\beta)}}{\beta-\gamma},
\]
it follows that
\[
\xi
=
\sigma_2\sigma_3\sqrt{-(\alpha+\gamma)(\alpha+\beta)} t.
\]
Because $\alpha+\gamma<0$ and $\alpha+\beta>0$, the quantity
$-(\alpha+\gamma)(\alpha+\beta)$ is strictly positive, so $\xi\neq 0$.
Hence these are genuine nontrivial conformal relative equilibria.

Moreover, along each of these rays both the Hamiltonian and the Casimir vanish.
Indeed,
\[
H\bigl(\mu_e^{\sigma_2,\sigma_3}(t)\bigr)
=
\frac{t^2}{2}\bigl(\alpha+\beta p^2+\gamma q^2\bigr)=0,
\]
and
\[
C\bigl(\mu_e^{\sigma_2,\sigma_3}(t)\bigr)
=
t^2\bigl(-1+p^2+q^2\bigr)=0.
\]
Thus the eight rays lie in the intersection $H=0$ and $C=0$, that is, on the
light cone in $\mathfrak{so}(2,1)^*$.

\textbf{Special case.}
Take
$\alpha=1$, $\beta=3$, $\gamma=-2$.
Then $-\beta<\alpha<-\gamma$ holds, and
\[
p=\frac{1}{\sqrt{5}},\qquad
q=\frac{2}{\sqrt{5}},\qquad
\sqrt{-(\alpha+\gamma)(\alpha+\beta)}=2.
\]
Therefore, the eight rays of nontrivial conformal relative equilibria are
\[
\mu_e^{\sigma_2,\sigma_3}(t)
=
t
\begin{pmatrix}
1\\
\sigma_2 \tfrac{1}{\sqrt{5}}\\[1mm]
\sigma_3 \tfrac{2}{\sqrt{5}}
\end{pmatrix},
\qquad t \in \mathbb{R} \setminus \{ 0 \} ,
\]
with corresponding velocities
$\xi=2\sigma_2\sigma_3t$.
Thus $\xi=2t$ when $\sigma_2$ and $\sigma_3$ have the same sign, and
$\xi=-2t$ when they have opposite signs. The cone $H=0$, its intersections
with the Casimir leaves $C=\mathrm{const}$, and the eight rays of conformal
relative equilibria are shown in Figure~\ref{fig:CRE-so(21)}.
\end{Example}

\section{The three-dimensional case: Bianchi classification}\label{section6}

The goal of this section is to give a complete list of real
three-dimensional Lie--Poisson structures that admit a homogeneous
Hamiltonian with nontrivial conformal relative equilibria.
By Proposition~\ref{prop:LP-criterion}, this is equivalent to
determining which three-dimensional Lie algebras contain a hyperbolic
element, i.e., an element whose adjoint map has a nonzero real
eigenvalue. Equivalently, one may ask for which algebras some coadjoint
operator has a nonzero eigenvector with nonzero real eigenvalue.
Since all real three-dimensional Lie algebras are classified by the
Bianchi scheme \cite{bianchi1897sugli,bianchi1928lezioni,ellis1969class},
this yields a definitive answer, which is stated and proved as
Proposition~\ref{prop:bianchi-cre} below.
In Section~\ref{subsec:structure-constants}, we recall the
Ellis--MacCallum form of the structure constants and explain how they
determine the Lie--Poisson tensor for each Bianchi type.
In Section~\ref{subsec:cre-bianchi}, we apply
Proposition~\ref{prop:LP-criterion} to each type in turn and identify
precisely which algebras admit nontrivial conformal relative equilibria.

\subsection{Structure constants and Lie--Poisson brackets}\label{subsec:structure-constants}

Every real three-dimensional Lie algebra is isomorphic to one
of the Bianchi types, including the one-parameter families VI$_h$ and VII$_h$ \cite{ellis1969class,ryan2015homogeneous}.
Following the scheme of Ellis and MacCallum \cite{ellis1969class}
(see also \cite{ryan2015homogeneous,yoshida2017rattleback}), the
structure constants of any three-dimensional Lie algebra can be written
in the form
\begin{equation}\label{eq:structure-constants}
 C^i{}_{jk}
 = \varepsilon_{jks} m^{si} + \delta^i_k a_j - \delta^i_j a_k,
\end{equation}
where $\varepsilon_{jks}$ denotes the Levi--Civita symbol: it is equal to $1$
for even permutations of $(1,2,3)$, equal to $-1$ for odd permutations, and
equal to $0$ whenever two indices are equal;
here $m=(m^{ij})$ is a $3\times3$ symmetric
matrix, and $a=(a_i)$ is a covector.
The nine Bianchi types and families correspond to different choices of $m$ and $a$,
as listed in Table~\ref{tab:bianchi-ma}.
The algebras split into two classes:
\begin{itemize}\itemsep=0pt
\item Class A ($a_i=0$): types I, II, $\mathrm{VI}_{-1}$,
$\mathrm{VII}_0$, VIII, IX.
\item Class B ($a_i\neq0$): types III, IV, V,
$\mathrm{VI}_{h\neq 0,-1}$, $\mathrm{VII}_{h\neq0}$.
\end{itemize}

In this parametrization, the value $h=0$ in the $\mathrm{VI}_h$
family is represented separately by type~III, while $h=-1$ is the
unimodular type $\mathrm{VI}_{-1}$.

\begin{table}[htbp]\small
\centering
\renewcommand{\arraystretch}{1.35}\small
\begin{tabular}{cclll}
\hline
Class & Type & Matrix $m$ & Vector $a=(a_1,a_2,a_3)$ & Notes \\
\hline
A & I & $0$ & $(0,0,0)$ & abelian; $\mathbb{R}^3$ \\
A & II & $\mathrm{diag}(1,0,0)$ & $(0,0,0)$ & nilpotent; Heisenberg \\
A & VI$_{-1}$& $-\alpha$ & $(0,0,0)$ & solvable\\
A & VII$_0$ & $\mathrm{diag}(-1,-1,0)$ & $(0,0,0)$ & solvable; $\mathfrak{e}(2)$ \\
A & VIII & $\mathrm{diag}(-1,1,1)$ & $(0,0,0)$ & simple; $\mathfrak{so}(2,1)\cong \mathfrak{sl}(2,\mathbb{R})$ \\
A & IX & $\mathrm{diag}(1,1,1)$ & $(0,0,0)$ & simple; $\mathfrak{so}(3)$ \\
\hline
B & III & $-\tfrac{1}{2}\alpha$ & $\bigl(0,0,-\tfrac{1}{2}\bigr)$ & solvable\\
B & IV & $\mathrm{diag}(1,0,0)$ & $(0,0,-1)$ & solvable \\
B & V & $0$ & $(0,0,-1)$ & solvable \\
B & $\mathrm{VI}_{h\neq 0,-1}$ & $\tfrac{1}{2}(h-1)\alpha$ & $\bigl(0,0,-\tfrac{1}{2}(h+1)\bigr)$ & solvable; non-unimodular family \\
B & VII$_{h\neq 0}$& $\mathrm{diag}(-1,-1,0)+\tfrac{h}{2}\alpha$ & $\bigl(0,0,-\tfrac{h}{2}\bigr)$ & solvable; non-unimodular family \\
\hline
\end{tabular}
\caption{Bianchi classification of real three-dimensional Lie algebras
after Ellis--MacCallum \cite{ellis1969class,ryan2015homogeneous}.
Here \smash{$\alpha=\bigl(\begin{smallmatrix}0&1&0\\1&0&0\\0&0&0\end{smallmatrix}\bigr)$}.
The vector $a$ is written explicitly to avoid notation ambiguities
(e.g., $a_i=-\delta^i_3$) commonly found in the literature.
The table follows the Ellis--MacCallum convention, in which some labels differ from other standard parametrizations of the Bianchi classes.}\label{tab:bianchi-ma}
\end{table}

Given the structure constants \eqref{eq:structure-constants}, the
associated Lie--Poisson bracket on $\mathfrak{g}^*\cong\mathbb{R}^3$
is
\begin{equation}\label{eq:LP-bracket}
 \{f,g\}(x)
 = \bigl\langle x, [\mathrm{d}f, \mathrm{d}g]\bigr\rangle
 = \sum_{i,j} \Pi_{ij}(x) \frac{\partial f}{\partial x_i}
 \frac{\partial g}{\partial x_j},
\end{equation}
where the Poisson tensor $\Pi$ is the antisymmetric matrix with entries
\begin{equation}\label{eq:poisson-matrix}
 \Pi_{ij}(x) = C^k_{ij} x_k.
\end{equation}
Note that $\Pi$ is \emph{linear} in $x$, a characteristic feature of Lie--Poisson structures.
Because sign conventions for $\mathfrak{so}(2,1)$ vary in
the literature, we compute its structure constants directly from~\eqref{eq:structure-constants}; in this case,
$m=\mathrm{diag}(-1,1,1)$ and $a=0$.

Substituting into \eqref{eq:structure-constants} gives
\[
C^1{}_{23}=-1,\qquad C^2{}_{31}=1,\qquad C^3{}_{12}=1,\qquad
C^1{}_{32}=1,\qquad C^2{}_{13}=-1,\qquad C^3{}_{21}=-1,
\]
and all other structure constants zero.
Hence
\[
\Pi(x)=
\begin{pmatrix}
0 & x_3 & -x_2\\
-x_3 & 0 & -x_1\\
x_2 & x_1 & 0
\end{pmatrix},
\]
which is the standard Lie--Poisson tensor of $\mathfrak{so}(2,1)$,
with Casimir $C(x) = -x_1^2 + x_2^2 +x_3^2$.

Carrying out the same computation for each Bianchi type yields the
Poisson tensors and Casimir functions listed in
Table~\ref{tab:bianchi-JP}.\footnote{The Poisson tensors in Table~\ref{tab:bianchi-JP} have been
computed directly from the Ellis--MacCallum structure constants~\eqref{eq:structure-constants}. This is useful because sign conventions
for the Bianchi Lie algebras vary in the literature; in particular,
the signs for types~V and~VIII are fixed here by the convention used in
Table~\ref{tab:bianchi-ma}.}


\begin{table}[ht]
\centering
\small
\renewcommand{\arraystretch}{1.0}
\setlength{\tabcolsep}{4pt}
\begin{tabular}{cllcc}
\hline
\noalign{\smallskip}
Type & Poisson tensor $\Pi$ & Casimir $C$ & Class & CRE? \\
\noalign{\smallskip}
\hline
\noalign{\smallskip}
I &
$\left(\begin{smallmatrix}0&0&0\\0&0&0\\0&0&0\end{smallmatrix}\right)$ &
$x_1,x_2,x_3$ & A & No \\[2pt]
II &
$\left(\begin{smallmatrix}0&0&0\\0&0&x_1\\0&{-x_1}&0\end{smallmatrix}\right)$ &
$x_1$ & A & No \\[2pt]
VI$_{-1}$ &
$\left(\begin{smallmatrix}0&0&x_1\\0&0&-x_2\\-x_1&{x_2}&0\end{smallmatrix}\right)$ &
$x_1 x_2$ & A & Yes \\[2pt]
VII$_0$ &
$\left(\begin{smallmatrix}0&0&x_2\\0&0&{-x_1}\\{-x_2}&x_1&0\end{smallmatrix}\right)$ &
$x_1^2+x_2^2$ & A & No \\[2pt]
VIII &
$\left(\begin{smallmatrix}0&x_3&{-x_2}\\{-x_3}&0&{-x_1}\\x_2&x_1&0\end{smallmatrix}\right)$ &
$-x_1^2+x_2^2+x_3^2$ & A & Yes \\[2pt]
IX &
$\left(\begin{smallmatrix}0&x_3&{-x_2}\\{-x_3}&0&x_1\\x_2&{-x_1}&0\end{smallmatrix}\right)$ &
$x_1^2+x_2^2+x_3^2$ & A & No \\[2pt]
\noalign{\smallskip}
\hline
\noalign{\smallskip}
III &
$\left(\begin{smallmatrix}0&0&x_1\\0&0&0\\-x_1&0&0\end{smallmatrix}\right)$ &
$x_2$ & B & Yes \\[2pt]
IV &
$\left(\begin{smallmatrix}0&0&{x_1}\\0&0&{x_1+x_2}\\-x_1&{-(x_1+x_2)}&0\end{smallmatrix}\right)$ &
$x_2 x_1^{-1}-\log x_1$ & B & Yes \\[2pt]
V &$
\left(\begin{smallmatrix}
0&0&x_1\\
0&0&x_2\\
-x_1&-x_2&0
\end{smallmatrix}\right)$ &$x_2x_1^{-1}$ &B & Yes \\[2pt]
$\mathrm{VI}_{h\neq 0,-1}$ &
$\left(\begin{smallmatrix}0&0&x_1\\0&0&hx_2\\{-x_1}&{-hx_2}&0\end{smallmatrix}\right)$ &
$x_2 x_1^{-h}$ & B & Yes \\[2pt]
VII$_{h\neq 0}$ $(|h|<2)$ &
$\left(\begin{smallmatrix}0&0&x_2\\0&0&{-x_1+hx_2}\\{-x_2}&x_1{-}hx_2&0\end{smallmatrix}\right)$ &
$C_{{\rm VII}_h}(x_1,x_2)$ & B & No \\
VII$_{h\neq 0}$ $(|h|\ge 2)$ &
$\left(\begin{smallmatrix}0&0&x_2\\0&0&{-x_1+hx_2}\\{-x_2}&x_1{-}hx_2&0\end{smallmatrix}\right)$ &
$C_{{\rm VII}_h}(x_1,x_2)$ & B & Yes \\
\noalign{\smallskip}
\hline
\end{tabular}
\caption{Lie--Poisson tensors, Casimir functions, and CRE existence for
the Bianchi types. For type VII$_{h\neq 0}$, the table is split by the CRE
regimes $0<|h|<2$ and $|h|\ge 2$; the Casimir $C_{{\rm VII}_h}$ has separate
expressions for $h<-2$, $h=2$, $-2<h<2$, $h=-2$, and $ h>2$; see
\cite{yoshida2017rattleback}. For types IV, V, $\mathrm{VI}_h$, and~$\mathrm{VII}_h$, the displayed
Casimir functions are understood locally on suitable coordinate
domains. For types IV, V, and~$\mathrm{VI}_h$, one may, for example,
take $x_1>0$; for type $\mathrm{VII}_h$, the appropriate local
expressions and domains depend on the parameter $h$.}
\label{tab:bianchi-JP}
\end{table}

\subsection{Conformal relative equilibria}\label{subsec:cre-bianchi}

We now apply Proposition~\ref{prop:LP-criterion} to determine which
Bianchi types admit a homogeneous Hamiltonian with nontrivial
conformal relative equilibria on $\mathfrak g^*$.
By Proposition~\ref{prop:LP-criterion}, this is equivalent to the
existence of $\zeta\in\mathfrak g$ such that the coadjoint operator
$\mathrm{ad}^*_\zeta\colon \mathfrak g^*\to\mathfrak g^*$ has a nonzero real
eigenvalue; equivalently, $\mathfrak g$ contains a hyperbolic element.

Following our conventions, type VI$_{-1}$ and type VII$_0$ are listed
separately, while VI$_h$ denotes the family with parameter
$h\neq 0,-1$ and VII$_h$ denotes the family with parameter $h\neq 0$.

For the Lie--Poisson structure on $\mathfrak g^*$, the Poisson tensor
defines at each point $x\in\mathfrak g^*$ a bundle map
$\Pi^\sharp_x\colon T_x^*\mathfrak g^*\to T_x\mathfrak g^*$. Under the canonical
identifications $T_x^*\mathfrak g^*\cong\mathfrak g$ and
$T_x\mathfrak g^*\cong\mathfrak g^*$, this map is precisely the
coadjoint action, $\Pi^\sharp_x(\zeta)=\mathrm{ad}^*_\zeta x$.
Thus, for each fixed $\zeta\in\mathfrak g$, we obtain a linear map
$A_\zeta\colon \mathfrak g^*\to\mathfrak g^*$ given by
$A_\zeta(x):=\mathrm{ad}^*_\zeta x=\Pi(x)\zeta$.
The following proposition determines, for each Bianchi type, whether
its Lie--Poisson structure admits a homogeneous Hamiltonian with
nontrivial conformal relative equilibria.

\begin{Proposition}[Bianchi types admitting nontrivial conformal relative equilibria]\label{prop:bianchi-cre}
A real three-dimensional Lie algebra $\mathfrak g$ endowed with its
Lie--Poisson structure admits a homogeneous Hamiltonian with a nontrivial
conformal relative equilibrium on $\mathfrak g^*$ if and only if
$\mathfrak g$ is of Bianchi type {\rm III}, {\rm IV}, {\rm V}, {\rm VI}$_{-1}$,
{\rm VI}$_h$ $(h\neq 0,-1)$, {\rm VII}$_h$ $(|h|\ge 2)$, or {\rm VIII}.

Equivalently, no such Hamiltonian exists for types {\rm I}, {\rm II}, {\rm VII}$_0$, {\rm IX},
and {\rm VII}$_h$ with $0<|h|<2$.
\end{Proposition}

\begin{proof}
We examine each Bianchi type in turn, using the Poisson tensors listed in
Table~\ref{tab:bianchi-JP}.

\textbf{Type I.}
Here $\Pi(x)=0$, so $A_\zeta=0$ for every $\zeta\in\mathfrak g$.
Thus all eigenvalues are zero, and no nontrivial conformal relative
equilibrium exists.

\textbf{Type II.}
From
\[
\Pi(x)=
\begin{pmatrix}
0&0&0\\
0&0&x_1\\
0&-x_1&0
\end{pmatrix}
\]
we obtain, for $\zeta=(\zeta_1,\zeta_2,\zeta_3)^{\mathsf T}$,
\[
A_\zeta x=\Pi(x)\zeta=
\begin{pmatrix}
0\\
\zeta_3 x_1\\
-\zeta_2 x_1
\end{pmatrix},
\qquad
A_\zeta=
\begin{pmatrix}
0&0&0\\
\zeta_3&0&0\\
-\zeta_2&0&0
\end{pmatrix}.
\]
This matrix is nilpotent, so its only eigenvalue is $0$. Hence type II
admits no nontrivial conformal relative equilibrium.

\textbf{Type III.}
Choose $\zeta=e_3$. Then
\[
A _{ e _3 } x = \Pi(x)e_3=
\begin{pmatrix}
x_1\\
0\\
0
\end{pmatrix},
\qquad
A_{e_3}=
\begin{pmatrix}
1&0&0\\
0&0&0\\
0&0&0
\end{pmatrix},
\]
whose eigenvalues are 1, 0, 0. Hence type III admits nontrivial
conformal relative equilibria.

\textbf{Type IV.}
Choose $\zeta=e_3$. Then
\[
A _{ e _3 } x =
\Pi(x)e_3=
\begin{pmatrix}
x_1\\
x_1+x_2\\
0
\end{pmatrix},
\qquad
A_{e_3}=
\begin{pmatrix}
1&0&0\\
1&1&0\\
0&0&0
\end{pmatrix}.
\]
Its eigenvalues are 1, 1, 0, so type IV admits nontrivial conformal
relative equilibria.

\textbf{Type V.}
Choose $\zeta=e_3$. Then
\[
A _{ e _3 } x =
\Pi(x)e_3=
\begin{pmatrix}
x_1\\
x_2\\
0
\end{pmatrix},
\qquad
A_{e_3}=
\begin{pmatrix}
1&0&0\\
0&1&0\\
0&0&0
\end{pmatrix},
\]
whose eigenvalues are $1,1,0$. Thus type V admits nontrivial
conformal relative equilibria.

\textbf{Type VI$_{-1}$.}
Choose $\zeta=e_3$. Then
\[
A _{ e _3 } x =
\Pi(x)e_3=
\begin{pmatrix}
x_1\\
-x_2\\
0
\end{pmatrix},
\qquad
A_{e_3}=
\begin{pmatrix}
1&0&0\\
0&-1&0\\
0&0&0
\end{pmatrix}.
\]
Its eigenvalues are $1$, $-1$, $0$. Hence type VI$_{-1}$ admits nontrivial
conformal relative equilibria.

\textbf{Type VI$_h$ $(h\neq 0,-1)$.}
Choose $\zeta=e_3$. Then
\[
A _{ e _3 } x =
\Pi(x)e_3=
\begin{pmatrix}
x_1\\
h x_2\\
0
\end{pmatrix},
\qquad
A_{e_3}=
\begin{pmatrix}
1&0&0\\
0&h&0\\
0&0&0
\end{pmatrix}.
\]
Its eigenvalues are $1$, $h$, $0$. Since $h\neq 0$, this gives a nonzero
real eigenvalue, so every type VI$_h$ with $h\neq 0,-1$ admits
nontrivial conformal relative equilibria.

\textbf{Type VII$_0$.} Let $\zeta=(\zeta_1,\zeta_2,\zeta_3)^{\mathsf T}\in\mathfrak g$. Then from
\[
\Pi(x)=
\begin{pmatrix}
0&0&x_2\\
0&0&-x_1\\
-x_2&x_1&0
\end{pmatrix}
\]
we get
\[
A_\zeta x= \Pi(x)\zeta=
\begin{pmatrix}
\zeta_3 x_2\\
-\zeta_3 x_1\\
\zeta_2 x_1-\zeta_1 x_2
\end{pmatrix},
\qquad
A_\zeta=
\begin{pmatrix}
0&\zeta_3&0\\
-\zeta_3&0&0\\
\zeta_2&-\zeta_1&0
\end{pmatrix}.
\]
Its characteristic polynomial is
\[
\chi_{A_\zeta}(\lambda)=\lambda(\lambda^2+\zeta_3^2),
\]
so the eigenvalues are $0$, $\pm i\zeta_3$. In particular, there is no
nonzero real eigenvalue, and therefore no nontrivial conformal relative
equilibrium.

\textbf{Type VII$_h$ $(h\neq 0)$.}
Let $\zeta=(\zeta_1,\zeta_2,\zeta_3)^{\mathsf T}\in\mathfrak g$. Then
\[
A_\zeta x=\Pi(x)\zeta=
\begin{pmatrix}
\zeta_3 x_2\\
-\zeta_3 x_1+h\zeta_3 x_2\\
\zeta_2 x_1-(\zeta_1+h\zeta_2)x_2
\end{pmatrix},
\]
so
\[
A_\zeta=
\begin{pmatrix}
0&\zeta_3&0\\
-\zeta_3&h\zeta_3&0\\
\zeta_2&-(\zeta_1+h\zeta_2)&0
\end{pmatrix}.
\]
Its characteristic polynomial is
\[
\chi_{A_\zeta}(\lambda)
=\lambda(\lambda^2-h\zeta_3\lambda+\zeta_3^2).
\]
If $\zeta_3=0$, then $\chi_{A_\zeta}(\lambda)=\lambda^3$, so all
eigenvalues are zero. If $\zeta_3\neq 0$, the two nonzero eigenvalues are
\[
\lambda_\pm=
\frac{h\zeta_3\pm\sqrt{\zeta_3^2(h^2-4)}}{2}.
\]
Hence $A_\zeta$ has real nonzero eigenvalues if and only if
$h^2-4\ge 0$. Therefore, $A_\zeta$ has a nonzero real eigenvalue for
some $\zeta\in\mathfrak g$ if and only if $|h|\ge 2$.

\textbf{Type VIII.}
Choose $\zeta=e_2$. Then
\[
A _{ e _2} x =
\Pi(x)e_2=
\begin{pmatrix}
x_3\\
0\\
x_1
\end{pmatrix},
\qquad
A_{e_2}=
\begin{pmatrix}
0&0&1\\
0&0&0\\
1&0&0
\end{pmatrix}.
\]
Its eigenvalues are $1$, $-1$, $0$. Therefore, type VIII admits nontrivial
conformal relative equilibria.

\textbf{Type IX.} Let $\zeta=(\zeta_1,\zeta_2,\zeta_3)^{\mathsf T}\in\mathfrak g$. Then from
\[
\Pi(x)=
\begin{pmatrix}
0&x_3&-x_2\\
-x_3&0&x_1\\
x_2&-x_1&0
\end{pmatrix}
\]
we obtain
\[
A_\zeta x=\Pi(x)\zeta=
\begin{pmatrix}
-x_2\zeta_3+x_3\zeta_2\\
x_1\zeta_3-x_3\zeta_1\\
-x_1\zeta_2+x_2\zeta_1
\end{pmatrix},
\qquad
A_\zeta=
\begin{pmatrix}
0&-\zeta_3&\zeta_2\\
\zeta_3&0&-\zeta_1\\
-\zeta_2&\zeta_1&0
\end{pmatrix}.
\]
This matrix is skew-symmetric, hence
\[
\chi_{A_\zeta}(\lambda)=
\lambda\bigl(\lambda^2+\zeta_1^2+\zeta_2^2+\zeta_3^2\bigr).
\]
Thus the spectrum consists of $0$ and a purely imaginary pair, so type
IX admits no nontrivial conformal relative equilibrium.

This proves the proposition.
\end{proof}

\appendix
\section{Appendix}\label{appendixA}

For a bivector field $\Pi\in\mathfrak{X}^2(M)$, the associated
\emph{sharp map} is the bundle map
$\Pi^\sharp\colon  T^*M\longrightarrow TM$
characterized by
$\langle \beta,\Pi^\sharp(\alpha)\rangle=\Pi(\beta,\alpha)$
for all $\alpha,\beta\in T^*M$.

Likewise, for a $2$-form $\omega\in\Omega^2(M)$, the associated
\emph{flat map} is the bundle map
$\omega^\flat\colon TM\longrightarrow T^*M$
characterized by
$\langle\omega^\flat(X),Y\rangle =\omega(X,Y)$
for all $X,Y\in TM$.

\begin{Lemma}[naturality of the Poisson bundle map]
\label{lemm:naturality-poisson}
Let $\Phi\colon M\to N$ be a diffeomorphism and let
$\Pi\in\mathfrak{X}^2(M)$. Then
$(\Phi_*\Pi )^\sharp=\Phi_*\circ\Pi^\sharp\circ\Phi^*$.
\end{Lemma}

\begin{proof}
Let $q\in N$, set $p=\Phi^{-1}(q)$, and let
$\alpha,\beta\in T_q^*N$. By the definition of the pushforward of a
bivector and the defining property of the sharp map,
\[
\bigl\langle
\beta,
(\Phi_*\Pi)^\sharp(\alpha)
\bigr\rangle
 =
(\Phi_*\Pi)(\beta,\alpha)
 =
\Pi(\Phi^*\beta,\Phi^*\alpha)
 =
\bigl\langle
\Phi^*\beta,
\Pi^\sharp(\Phi^*\alpha)
\bigr\rangle.
\]
By the defining property of the pullback of a one-form,
\[
\bigl\langle
\Phi^*\beta,
\Pi^\sharp(\Phi^*\alpha)
\bigr\rangle
=
\bigl\langle
\beta,
\Phi_*\bigl(\Pi^\sharp(\Phi^*\alpha)\bigr)
\bigr\rangle.
\]
Since this holds for every $\beta\in T_q^*N$, it follows that
\[
(\Phi_*\Pi)^\sharp(\alpha)
=
\bigl(\Phi_*\circ\Pi^\sharp\circ\Phi^*\bigr)(\alpha).\tag*{\qed}
\]\renewcommand{\qed}{}
\end{proof}

\begin{Lemma}[naturality of the symplectic bundle map]\label{lemm:naturality-symplectic}
Let $\Phi\colon M\to N$ be a diffeomorphism and let
$\omega\in\Omega^2(N)$. Then
$(\Phi^*\omega)^\flat=\Phi^*\circ\omega^\flat\circ\Phi_*$.
\end{Lemma}

\begin{proof}
Let $p\in M$ and let $X,Y\in T_pM$. By the definition of the pullback
of a $2$-form,
\[
\bigl\langle
(\Phi^*\omega)^\flat(X),
Y
\bigr\rangle
=
(\Phi^*\omega)(X,Y)
=
\omega(\Phi_*X,\Phi_*Y)
=
\bigl\langle
\omega^\flat(\Phi_*X),
\Phi_*Y
\bigr\rangle.
\]
Using the defining property of the pullback of a one-form, we obtain
\[
\bigl\langle
\omega^\flat(\Phi_*X),
\Phi_*Y
\bigr\rangle
=
\bigl\langle
\Phi^*\bigl(\omega^\flat(\Phi_*X)\bigr),
Y
\bigr\rangle.
\]
Since this holds for every $Y\in T_pM$, it follows that
\[
(\Phi^*\omega)^\flat(X)
=
\bigl(\Phi^*\circ\omega^\flat\circ\Phi_*\bigr)(X).\tag*{\qed}
\]\renewcommand{\qed}{}
\end{proof}

\subsection*{Acknowledgements}
The author is grateful to the anonymous referees for their careful reading and insightful comments, which led to significant improvements throughout the manuscript.

\pdfbookmark[1]{References}{ref}
\LastPageEnding

\end{document}